\documentclass[11pt,a4paper,reqno]{amsart}
\usepackage{multirow}
\usepackage[centertags]{amsmath}
\usepackage{amsfonts}
\usepackage{amssymb}
\usepackage{amsthm}
\usepackage{newlfont}
\usepackage{a4}
\usepackage{enumerate}
\usepackage[pdftex]{graphicx,color}
\usepackage{fancyref}
\usepackage{cite}

\theoremstyle{definition}

\theoremstyle{remark}

\newcounter{prop}[defn]
\renewcommand*{\theprop}{\thesubsection.\arabic{prop}}

\newenvironment{tvr}{%
  \refstepcounter{prop}%
	\vspace{5pt}
  \paragraph{\textbf{Proposition~\theprop}}%
}

\newenvironment{thm}{%
  \refstepcounter{prop}%
		\vspace{5pt}
  \paragraph{\textbf{Theorem~\theprop}}%
}

\newenvironment{example}{%
  \refstepcounter{prop}%
			\vspace{5pt}
  \paragraph{\textbf{\textit{Example}}\textbf{~\theprop}}%
}

\usepackage{bbm}

\theoremstyle{definition}

\newcommand{\hartcoefficient}{l}
\newcommand{\cas}{\mathrm{cas}}
\newcommand{\Hart}{\zeta}
\newcommand{\Rere}{\mathrm{Re}\,}
\newcommand{\Imim}{\mathrm{Im}\,}
\newcommand{\e}{\mathrm{e}}
\theoremstyle{remark}
\newcommand{\one}{\mathfrak{1}}

\newcommand{\x}{a}

\newcommand{\si}{\sigma}

\newcommand{\aff}{\mathrm{aff}}

\newcommand{\dual}{^\vee}

\newcommand{\wdual}[1]{\omega^\vee_{#1}}

\newcommand{\wvector}[1]{\omega_{#1}}
\newcommand{\R}{\mathbb{R}}

\newcommand{\scalar}[2]{ \left\langle #1,#2 \right\rangle}
\newcommand{\Z}{\mathbb{Z}}
\newcommand{\N}{\mathbb{N}}

\newcommand{\Stab}{\text{Stab}}
\newcommand{\set}[2]{\left\{#1 \,\mid \, #2 \right\}}

\newcommand{\la}{\lambda}
\newcommand{\om}{\omega}
\newcommand{\setb}[2]{\left\{#1 \, \mid\, #2 \right\}}
\newcommand{\ep}{\epsilon}
\newcommand{\im}{\mathrm{i}}

\newcommand{\map}{\rightarrow}
\newcommand{\Int}{\operatorname{I}}
\newcommand{\Inth}{\operatorname{Ih}}
\renewcommand{\int}{\operatorname{int}}

\newcommand{\q}{\quad}

\renewcommand{\epsilon}{\varepsilon}

\newcommand{\al}{\alpha}

\renewcommand{\rho}{\varrho}

\newcommand{\Com}{{\mathbb C}}

\newcommand{\setm}[2]{\left\{#1 \, \big|\, #2 \right\}}
\newcommand{\abs}[1]{\left\vert#1\right\vert}
\newcommand{\wt}{\widetilde}
\newcommand{\wh}{\widehat}

\newcommand{\sca}[2]{\langle #1,\, #2\rangle}
\newcommand{\comb}[2]{\begin{pmatrix}
     #1\\
     #2
  \end{pmatrix}}

\usepackage{ulem}
\begin{document}

\title[On $E-$Discretization of Tori of Lie Groups: II]
{On $E-$Discretization of Tori of Compact Simple Lie Groups: II}

\author[J. Hrivn\'{a}k]{Ji\v{r}\'{\i} Hrivn\'{a}k$^{1}$}
\author[M. Jur\'{a}nek]{Michal Jur\'{a}nek$^{1}$}

%%%%%%%%%%%%%%%%%%%%%%%%%%%%%%%%%%%%%%%%%%%%%%%%%%%%%%%%%%%%%%%%%%
%%%%%%%%%%%%%%%%%%%%%%%%%%%%%%%%%%%%%%%%%%%%%%%%%%%%%%%%%%%%%%%%%%
\begin{abstract}\small
Ten types of discrete Fourier transforms of Weyl orbit functions are developed. Generalizing one-dimensional cosine, sine and exponential, each type of the Weyl orbit function represents an exponential symmetrized with respect to a subgroup of the Weyl group. Fundamental domains of  even affine and dual even affine Weyl groups, governing the argument and label symmetries of the even orbit functions, are determined. The discrete orthogonality relations are formulated on finite sets of points from the refinements of the dual weight lattices. Explicit counting formulas for the number of points of the discrete transforms are deduced.  Real-valued Hartley orbit functions are introduced and all ten types of the corresponding discrete Hartley transforms are detailed.       
\end{abstract}

\maketitle
{\vspace{-10pt}
\noindent
$^1$ Department of Physics,
Faculty of Nuclear Sciences and Physical Engineering, Czech
Technical University in Prague, B\v{r}ehov\'a~7, 115 19 Prague 1, Czech
Republic
}
\vspace{10pt}

\noindent
\textit{E-mail:} jiri.hrivnak@fjfi.cvut.cz, michal.juranek@fjfi.cvut.cz

\date{\today}
%%%%%%%%%%%%%%%%%%%%%%%%%%%%%%%%%%%%%%%%%%%%%%%%%%%%%%%%%%%%%%%%%%
\section{Introduction}
%%%%%%%%%%%%%%%%%%%%%%%%%%%%%%%%%%%%%%%%%%%%%%%%%%%%%%%%%%%%%%%%%%

The purpose of this paper is to complete and extend the discrete Fourier analysis of Weyl-orbit functions from \cite{discliegrI,discliegrII,discliegrE}. The discrete Fourier calculus of all ten types of orbit functions with symmetries inherited from all four types of even Weyl groups is unified in full generality. The real-valued versions of the functions and transforms are also developed by modifying the exponential kernels of orbit functions to their Hartley alternatives \cite{Brac}.      

The ten types of complex orbit functions of Weyl groups and their even subgroups stem from the study of standard symmetric and antisymmetric orbit sums \cite{Bour} as special functions \cite{KP1,KP2}. Another set of special functions, the $E-$functions, is generated by symmetrizing exponentials over even subgroups of the Weyl groups. These normal even subgroups consist of Weyl group elements with unity determinant only and their corresponding invariant $E-$functions retain Fourier-like analysis properties. The discrete Fourier calculus of  
(anti)symmetric orbit functions and $E-$functions is developed in \cite{discliegrI,discliegrE}. For root systems with two lengths of roots, additional two sets of short and long orbit functions are generated utilizing the concept of short and long sign homomorphisms \cite{SsSlcub}. The discrete Fourier analysis of these short and long orbit functions is detailed in \cite{discliegrII}. 

The short and long sign homomorphisms induce two types of associated even Weyl groups which in turn produce five more types of $E-$functions. So far, their specific boundary behaviour, continuous and discrete Fourier transforms are studied only for two-variable cases in \cite{Ef2d}. These five types of functions are obtained by combinations of the given type of even Weyl group and one additional sign homomorphism. Moreover, it appears that all ten types of orbit functions, including the original (anti)symmetric orbit sums, are  similarly determined by the choice of two sign homomorphisms. The first sign homomorphism of the given orbit function rules its (anti)symmetry properties and the second indicates its underlying symmetry group. Such consolidated generalizing approach not only allows the deduction of the five remaining types of discrete Fourier transforms of even orbit functions in the present paper, but also incorporates all previously studied cases from \cite{discliegrI,discliegrII,discliegrE,Ef2d}. Another direction in generalization of the Fourier calculus of Weyl-orbit functions is motivated by the growing ubiquity of a real-valued version of the Fourier transform, the Hartley transform \cite{Poul}.

Since introduction of the discrete version of the Hartley transform in \cite{Brac}, both continuous and discrete Hartley transforms form fully equivalent real-valued variants of the standard Fourier transforms \cite{Poul}. As alternatives to complex Fourier transforms, these transforms together with their 2D and 3D versions found applications in many fields including signal processing \cite{Paras, Pusch}, pattern recognition \cite{Bhar}, geophysics \cite{Kuhl}, measurement \cite{Sun} and optics \cite{Liu}. In the context of Weyl-orbit functions and their corresponding transforms, the Hartley transforms have not yet been studied. Replacing exponential kernel as in original 1D Hartley transform yields novel families of real-valued special functions of Weyl groups, which inherit (anti)symmetry properties as well as discrete orthogonality relations from the original Weyl-orbit functions. The resulting generalized Hartley transforms together with the original ten types of Weyl-orbit functions offer, especially in 2D and 3D, richer options and application potential due to greater variability of domain shapes and boundary behaviour. 

The original (anti)symmetric Weyl-orbit functions constitute an integral part of conformal field theories with Lie group symmetries. Applications of other types of Weyl-orbit functions and their associated discrete transforms in this field are investigated in \cite{HW1,HW2}. One of the crucial properties of all ten types of Weyl-orbit functions is existence of product-to-sum decomposition formulas. These formulas allow the generalization of the classical mechanics model of eigenvibrations of a string with beads to their corresponding multidimensional versions generated by Weyl group symmetries and their inherent invariant lattices. Such mechanical models are intensively studied in connection with the graphene material \cite{KaHa,CsTi} and form stepping stones for their quantum field versions \cite{DrSa}. Dispersion relations of these models are commonly derived assuming solutions in exponential form while imposing periodic boundary conditions. Generalizing cosine and sine standing waves solutions of 1D beaded string with von Neumann and Dirichlet conditions, respectively, the ten types of Weyl-orbit functions potentially represent solutions with greatly expanded collection of boundary conditions.    

The paper is organized as follows. In Section 2, properties of  crystallographic root systems and corresponding affine Weyl groups are reviewed. In Section 3, necessary aspects of four types of affine even Weyl groups and their fundamental domains are derived.  The ten types of orbit functions together with their Hartley variants are introduced in Section 4. Discrete orthogonality of orbit functions  and the resulting discrete transforms of both complex and Hartley types are deduced in Section 5. The counting formulas for the number of points in the finite grid fragments are located in the Appendix.

%%%%%%%%%%%%%%%%%%%%%%%%%%%%%%%%%%%%%%%%%%%%%%%%%%%%%%%%%%%%%%%%%%
\section{Properties of crystallographic root systems}\setcounter{prop}{0}
%%%%%%%%%%%%%%%%%%%%%%%%%%%%%%%%%%%%%%%%%%%%%%%%%%%%%%%%%%%%%%%%%%

%%%%%%%%%%%%%%%%%%%%%%%%%%%%%%%%%%%%%%%%%%%%%%%%%%%%%%%%%%%%%%%%%%
\subsection{Definitions and notations}\setcounter{prop}{0}\
%%%%%%%%%%%%%%%%%%%%%%%%%%%%%%%%%%%%%%%%%%%%%%%%%%%%%%%%%%%%%%%%%%

Throughout the paper, the notation from \cite{discliegrI,discliegrE} is adopted. To the simple Lie algebra of rank $n$ of the compact simple Lie group corresponds the root system $\Pi$ \cite{BB,Bour,H2}. The set of simple roots $\Delta=\{\al_1,\dots,\al_n\}$ of $\Pi$ spans the Euclidean space $\R^n$, with the symbol $\sca{\,}{\,}$ denoting its scalar product. The set of simple roots determines partial ordering $\leq$ on $\R^n$; for $\la,\nu \in \R^n$ it holds that  $\nu\leq \la$ if and only if $\la-\nu = k_1\al_1+\dots+ k_n \al_n$ with $k_i \in \Z^{\geq 0}$ for all $i\in \{1,\dots,n\}$. Note that the notions of the root system $\Pi$ and its inherent set of simple roots $\Delta$ are also developed independently on Lie theory and the sets $\Delta$ and $\Pi$ corresponding to compact simple Lie groups are called crystallographic \cite{H2}. There are two types of sets of simple roots --- the first type with roots of only one length, denoted by $A_n, \, n\geq 1$, $D_n, \, n\geq 4$, $E_6$, $E_7$, $E_8$,  and the second type
with two different lengths of roots, denoted $B_n,\, n\geq 3$, $C_n,\, n\geq 2$, $G_2$ and $F_4$. For root systems with two different lengths of roots each set of simple roots consists of short simple roots $\Delta_s$ and long simple roots $\Delta_l$, i.e. the following disjoint decomposition is given,
\begin{equation}\label{sl}
\Delta=\Delta_s\cup\Delta_l.
\end{equation}
The following notation of the standard objects \cite{BB,H2}, which are induced by the set $\Delta\subset \Pi$, is used:
\begin{itemize}
\item the highest root $\xi \in \Pi$ with respect to the partial ordering $\leq$ restricted on $\Pi$,
\item the marks $m_1,\dots,m_n\in \N$ of the highest root $\xi=m_1\al_1+\dots+m_n\al_n$ together with $m_0=1$,
\item
the Coxeter number $m=m_0+m_1+\dots+m_n$,
\item
the root lattice $Q=\Z\al_1+\dots+\Z\al_n $,
\item
the $\Z$-dual lattice to $Q$,
\begin{equation*}
 P^{\vee}=\set{\om^{\vee}\in \R^n}{\sca{\om^{\vee}}{\al}\in\Z,\, \forall \al \in \Delta}=\Z \om_1^{\vee}+\dots +\Z \om_n^{\vee},
 \end{equation*}
with 
\begin{equation}\label{aldom} 
\sca{\al_i}{ \om_j^{\vee}}=\delta_{ij},
\end{equation}
\item
the dual root lattice $Q^{\vee}=\Z \al_1^{\vee}+\dots +\Z \al^{\vee}_n$, where 
\begin{equation}\label{aldual}
\al^{\vee}_i=\frac{2\al_i}{\sca{\al_i}{\al_i}},\q i\in \{1,\dots,n\},
\end{equation}
\item the dual marks $m^{\vee}_1, \dots ,m^{\vee}_n$ of the highest dual root $\eta= m_1^{\vee}\al_1^{\vee} + \dots + m_n^{\vee} \al_n^{\vee}$ together with $m_0^\vee=1$; the marks and the dual marks are summarized in Table 1 in \cite{discliegrI},
\item the $\Z$-dual lattice to $Q^\vee$
\begin{equation*}
 P=\set{\om\in \R^n}{\sca{\om}{\al^{\vee}}\in\Z,\, \forall \al^{\vee} \in Q^\vee}=\Z \om_1+\dots +\Z \om_n,
\end{equation*}
\item 
the Cartan matrix $C$ with elements 
\begin{equation*}
 C_{ij}=\sca{\al_i}{\al^{\vee}_j},
\end{equation*}
\item the index of connection $c$ of $\Pi$ equal to the  determinant of the Cartan matrix $C$,
\begin{equation}\label{Center}
 c=\det C.
\end{equation}
\end{itemize}

%The set of simple dual roots $\Delta^{\vee}=\{\al^{\vee}_1,\dots , \al^{\vee}_n\}$, with the dual roots given by \eqref{aldual}, is also a system of simple roots of the dual root system $W\Delta^{\vee}$. The system$\Delta^{\vee}$ also spans Euclidean space $\R^n$ and has analogous properties as the root system $\Delta.$The highest dual root $\eta\in W\Delta^{\vee}$  can be written as$$\eta\equiv -\al_0^{\vee}= m_1^{\vee}\al_1^{\vee} + \dots + m_n^{\vee} \al_n^{\vee}$$ and the coefficients $m^{\vee}_j$ are called the dual marks.
%%%%%%%%%%%%%%%%%%%%%%%%%%%%%%%%%%%%%%%%%%%%%%%%%%%%%%%%%%%%%%%%%%
%%%%%%%%%%%%%%%%%%%%%%%%%%%%%%%%%%%%%%%%%%%%%%%%%%%%%%%%%%%%%%%%%%
\subsection{Weyl group and affine Weyl group}\setcounter{prop}{0}\
%%%%%%%%%%%%%%%%%%%%%%%%%%%%%%%%%%%%%%%%%%%%%%%%%%%%%%%%%%%%%%%%%%
%%%%%%%%%%%%%%%%%%%%%%%%%%%%%%%%%%%%%%%%%%%%%%%%%%%%%%%%%%%%%%%%%%

The properties of Weyl groups and affine Weyl groups can be found for example in \cite{BB,H2}.
The finite Weyl group $W$ is generated by $n$ reflections $r_\al$, $\al\in\Delta$, in $(n-1)$-dimensional `mirrors' orthogonal to simple roots intersecting at the origin:
\begin{equation*}%\label{Weyl}
r_{\al_i}\x\equiv r_i \x=\x-\frac{2\sca{\x}{\al_i} }{\sca{\al_i}{\al_i}}\al_i\,,
\qquad \x\in\R^n\,.
\end{equation*}
The infinite affine Weyl group $W^{\mathrm{aff}}$ is the semidirect product of the Abelian group of translations $Q^\vee$ and of the Weyl group $W$
\begin{equation}\label{semiW}
 W^{\mathrm{aff}}= Q^\vee \rtimes W.
\end{equation}
Therefore, for each $w^\aff \in W^\aff$ there exists a unique $w \in W$ and a unique $q\dual \in Q\dual$ such that for each $\x \in \R^n$ it holds that
\begin{equation}\label{semiaff}
w^\aff \x = w\x +q\dual.
\end{equation}
 The retraction homomorphism $\psi : W^{\aff} \map W$ of the semidirect product \eqref{semiW} is given for the element of the form \eqref{semiaff} as
\begin{equation}\label{retract}
\psi (w^\aff) = w.
\end{equation}
Equivalently, $W^{\mathrm{aff}}$ is a Coxeter group generated by $n$ reflections $r_i$ and an affine reflection $r_0$ given as
\begin{equation*}
r_0 \x=r_\xi \x + \frac{2\xi}{\sca{\xi}{\xi}}\,,\qquad
r_{\xi}\x=\x-\frac{2\sca{\x}{\xi} }{\sca{\xi}{\xi}}\xi\,,\qquad \x\in\R^n\,.
\end{equation*}
The set of $n$ reflections $r_i$ together with the affine reflection $r_0$ is denoted by 
$$R=\{r_0, r_1,\dots, r_n\}$$
and the standard relations of the Coxeter group $W^{\mathrm{aff}}$
\begin{equation}\label{r2}
r^2=1, \q r\in R,
\end{equation}
hold for all its generators.

The fundamental region $F \subset \R^n$ of $W^{\mathrm{aff}}$ can be chosen as a simplex
$$F =\setb{\x\in \R^n}{\sca{\x}{\al}\geq 0,\forall \al\in\Delta,\sca{\x}{\xi}\leq 1  } $$ 
or, equivalently, as the convex hull of the points $\left\{ 0, \frac{\om^{\vee}_1}{m_1},\dots,\frac{\om^{\vee}_n}{m_n} \right\}$:
\begin{align}\label{deffun}
  F &=\setb{y_1\om^{\vee}_1+\dots+y_n\om^{\vee}_n}{y_0,\dots, y_n\in \R^{\geq 0}, \, y_0+y_1 m_1+\dots+y_n m_n=1  }.
\end{align}
Three essential properties of the fundamental region $F$ of $W^{\mathrm{aff}}$ are the following:
\begin{enumerate}
\item For any $a\in \R^n,$ there exists $a'\in F$ and $w^{\mathrm{aff}}\in W^{\mathrm{aff}},$  such that
\begin{equation}\label{fun1}
 a=w^{\mathrm{aff}}a'.
\end{equation}

\item If $a,a'\in F$ and $a'=w^{\mathrm{aff}}a$, $w^{\mathrm{aff}}\in W^{\mathrm{aff}},$ then
\begin{equation}\label{fun2}
 a'=a=w^{\mathrm{aff}}a.
\end{equation}

\item
Consider a point
$a=y_1\om^{\vee}_1+\dots+y_n\om^{\vee}_n\in F$, such that
$y_0+y_1 m_1+\dots+y_n m_n=1$.
The isotropy group
 \begin{equation}\label{stab}
 \mathrm{Stab}_{W^{\mathrm{aff}}}(a) = \setb{w^{\mathrm{aff}}\in W^{\mathrm{aff}}}{w^{\mathrm{aff}}a=a}
\end{equation}
of point $a$ is trivial, if $a\in \mathrm{int}(F)$, where $\mathrm{int}(F)$ denotes the interior of $F$, i.e. all $y_i>0$, $i=0,\dots,n$. Otherwise, if $a \notin \mathrm{int}(F)$, the group $\mathrm{Stab}_{W^{\mathrm{aff}}}(a)$ is finite and generated by such $r_i$ for which $y_i=0$, $i=0,\dots,n$.
\end{enumerate}

%%%%%%%%%%%%%%%%%%%%%%%%%%%%%%%%%%%%%%%%%%%%%%%%%%%%%%%%%%%%%%%%%%
%%%%%%%%%%%%%%%%%%%%%%%%%%%%%%%%%%%%%%%%%%%%%%%%%%%%%%%%%%%%%%%%%%
\subsection{Group of sign homomorphisms}\setcounter{prop}{0}\
%%%%%%%%%%%%%%%%%%%%%%%%%%%%%%%%%%%%%%%%%%%%%%%%%%%%%%%%%%%%%%%%%%
%%%%%%%%%%%%%%%%%%%%%%%%%%%%%%%%%%%%%%%%%%%%%%%%%%%%%%%%%%%%%%%%%%

Any homomorphism $\sigma$ from $W$ to the multiplicative group $\{1,-1\} $ is called a sign homomorphism \cite{discliegrII}. 
Two standard choices of sign homomorphisms are the trivial homomorphism and the determinant denoted as
\begin{alignat*}{2}
& \sigma^e (w) &=& \det(w),\\ 
& \one (w) &=& 1.
\end{alignat*} 

The sign homomorphisms $\sigma^l $ and $\sigma^s$ are defined on the set of generators $\set{r_\alpha} {\alpha \in \Delta}$ of $W$ as
\begin{alignat*}{1}
 \sigma^s (r_\alpha) &= \left\{ 
  \begin{array}{l l}
    -1 & \quad \text{if $\alpha \in \Delta_s, $}\\
     1 & \quad \text{otherwise,}
  \end{array} \right. \\
	\sigma^l (r_\alpha) &= \left\{ 
  \begin{array}{l l}
    -1 & \quad \text{if $\alpha \in \Delta_l, $}\\
     1 & \quad \text{otherwise.}
  \end{array} \right. 
\end{alignat*}
All sign homomorphisms form an abelian group $\Sigma$ with the pointwise operation $\cdot$
defined for $w \in W$ as
$$(\si_1 \cdot \si_2)(w)=\si_1(w)\si_2 (w). $$
The set of sign homomorphisms $\Sigma$ of a root system $\Delta$ with two different lengths of roots contains only four elements \cite{discliegrII}, i.e.
$$\Sigma = \{ \one,\, \si^e,\,\si^s,\,\si^l\}$$
and the group $\Sigma$ is in this case isomorphic to the Klein four-group.

The set of 'negative' generators from $R$ with respect to the sign homomorphism $\si$ is denoted by $R^\si$,
\begin{align}\label{Rsigma}
R^\sigma&=\setb{r\in R}{ \sigma \circ \psi \left( r\right)=-1 }.
\end{align}
This set is used in to decompose the fundamental domain $F$ into two disjoint parts. One of these parts is the subset $H^\si$ of boundaries of $F$,   
\begin{alignat*}{1}
H^\sigma &= \set{\x \in F}{(\exists r \in R^\sigma) (r\x=\x)}
\end{alignat*}
and the second set $F^\sigma\subset F$ is given by  
\begin{align}\label{Fsigma}
F^\sigma&=\setm{a\in F}{\sigma \circ \psi \left(\mathrm{Stab}_{W^{\mathrm{aff}}}(a)\right)=\{1\} }
\end{align}

The domains $F^s$, $R^s$ and $H^s$ from \cite{discliegrII} correspond in this notation to the sets $F^{\si^s}$, $R^{\si^s}$ and $H^{\si^s}$ and an analogous correspondence is valid for the long versions of these sets.
For the two standard cases of sign homomorphism $\si = \one, \si^e$ of $R^{\sigma}$ one obtains that $R^\one =\emptyset$ and $R^{\sigma^e} =R$. For the root systems with two lengths of roots it holds that $R= R^{\sigma^s}\cup R^{\sigma^l}$ and the subsets of generators $R^{s}$ and $R^{l}$ are summarized in Table 1 in \cite{discliegrII}. 
Moreover, the sets $F^{\sigma}$ and $H^{\sigma}$ are special cases of the sets $F^{\sigma}(\rho)$ and $H^{\sigma}(\rho)$ from \cite{CzHr}. Specializing Proposition 2.7 from \cite{CzHr}, the following set equality holds
\begin{equation}\label{Fsigma2}
F^\sigma=F\setminus H^\sigma,
\end{equation} 
which yields the disjoint decomposition of the fundamental domain $F$
\begin{equation}\label{decompF}
F=F^\sigma \cup H^\sigma
\end{equation}
together with crucial reformulation of the form of the boundary set $H^\sigma$,
\begin{align}\label{Hsigma}
H^\sigma&=\setm{a\in F}{\sigma \circ \psi \left(\mathrm{Stab}_{W^{\mathrm{aff}}}(a)\right)=\{\pm 1\} }.
\end{align}

%%%%%%%%%%%%%%%%%%%%%%%%%%%%%%%%%%%%%%%%%%%%%%%%%%%%%%%%%%%%%%%%%%
%%%%%%%%%%%%%%%%%%%%%%%%%%%%%%%%%%%%%%%%%%%%%%%%%%%%%%%%%%%%%%%%%%

\section{Affine even Weyl groups}\setcounter{prop}{0}
%%%%%%%%%%%%%%%%%%%%%%%%%%%%%%%%%%%%%%%%%%%%%%%%%%%%%%%%%%%%%%%%%%
%%%%%%%%%%%%%%%%%%%%%%%%%%%%%%%%%%%%%%%%%%%%%%%%%%%%%%%%%%%%%%%%%%
\subsection{Fundamental domains}\setcounter{prop}{0}\

Kernels of the non-trivial sign homomorphisms of a given Weyl group $W$ form normal subgroups $W^{\sigma}\subset W$ known as even Weyl groups \cite{KP3},
$$ W^{\sigma} \equiv \set {w \in W} { \sigma(w)=1}.$$ 
The corresponding affine even Weyl groups are the kernels of the expanded sign homomorphisms $\si \circ \psi$ or, equivalently, the semidirect products of the group of translations $Q\dual$ and the even Weyl groups $W^\si$,
$$ W_{\sigma}^{\mathrm{aff}} \equiv  \set {w^\aff \in W^\aff} { \sigma \circ \psi(w^\aff)=1} =Q^\vee \rtimes W^\sigma. $$
The fundamental domains of the even affine Weyl groups $W_\si^{\mathrm{aff}}$ are determined in the following proposition.

\begin{tvr}\label{tvrfund}
For any $r_\si \in R^\si$, the set $F\cup r_\si F^\si$ is  a fundamental domain of $W_\si^{\aff}$ satisfying the following conditions.
\begin{enumerate}
\item  For any $a\in \R^n,$ there exist $a'\in F\cup r_\sigma  F^\sigma$ and $w^{\mathrm{aff}}_\si\in W_\si^{\mathrm{aff}},$  such that
\begin{equation}\label{fun1s}
 a=w_\si^{\mathrm{aff}}a'.
\end{equation}
\item  If $a,a'\in F\cup r_\sigma  F^\si $ and $a'=w_\si^{\mathrm{aff}}a$, $w_\si^{\mathrm{aff}}\in W_\si^{\mathrm{aff}},$ then
\begin{equation}\label{FD2}
 a'=a=w_\si^{\mathrm{aff}}a.
\end{equation}
\item Consider a point $\x \in F\cup r_\si  F^\si.$ If $\x \in F^\si$ or $\x \in r_\si F^\si,$ then it holds that  
\begin{equation}\label{calc1} 
\mathrm{Stab}_{W^{\mathrm{aff}}_\si}(\x) =\mathrm{Stab}_{W^{\mathrm{aff}}}(\x),
\end{equation}
if $\x \in H^\si,$ then it holds that
$$\abs{\mathrm{Stab}_{W^{\mathrm{aff}}}(\x)} = 2\abs{\mathrm{Stab}_{W^{\mathrm{aff}}_\si}(\x)}.$$
\end{enumerate}
\end{tvr}
\begin{proof}\

\begin{enumerate}
\item Consider a point $a \in \R^n$. Due to \eqref{fun1}, there exist $a'' \in F$ and $w^\aff \in W^\aff$ such that $a=w^\aff a''.$ If $\si \circ \psi(w^\aff)=1$ then $w^\aff \in W_\si^\aff$; setting $a'=a'' $ and $w_\si^{\mathrm{aff}} = w^{\mathrm{aff}} $, condition \eqref{fun1s} is satisfied. On the other hand, suppose that $\sigma \circ \psi(w^\aff)=-1$.
Two distinct possibilities for $a'' \in F$ are determined by decomposition \eqref{decompF}. 
If $a'' \in F^\si$ then $r_\si a'' \in r_\si F^\si$ and $\si \circ \psi(w^\aff r_\si)=1$, i.e. $w^\aff r_\si \in W_\si^\aff $. Setting  $a'=r_\si a'' $, $w_\si^{\mathrm{aff}} = w^{\mathrm{aff}}r_\si $ and taking into account \eqref{r2}, condition \eqref{fun1s} is satisfied.
If $a'' \in H^\si,$ then there exists $r \in R^\si$ such that $a'' = r a''$ and $\si \circ \psi(w^\aff r)=1$, i.e. $w^\aff r \in W_\si^\aff $. Setting  $a'= a'' $, $w_\si^{\mathrm{aff}} = w^{\mathrm{aff}}r $, condition \eqref{fun1s} is satisfied again.
\item Suppose $a,a'\in F\cup r_\sigma  F^\si $ are such that $a'=w_\si^{\mathrm{aff}}a$ with $w_\si^{\mathrm{aff}}\in W_\si^{\mathrm{aff}}$. Note that the sets $F$ and $r_\si  F^\si$ are disjoint; indeed, take any $b\in F\cap r_\si  F^\si$ and consider $b'\in F^\si$ such that $r_\si b' = Ąb$. Then from \eqref{fun2} follows that $b'=b$ and thus $r_\si \in \mathrm{Stab}_{W^{\mathrm{aff}}}(b)$. Then it holds that $\sigma \circ \psi \left(\mathrm{Stab}_{W^{\mathrm{aff}}}(b)\right)=\{\pm 1\}$ and this violates definition \eqref{Fsigma}. 
Thus $F\cup r_\si  F^\si$ consists of two disjoint parts $F$ and $r_\si F^\si$ and the following cases are possible.
\begin{enumerate}
\item $a, a'\in F$. It follows immediately from \eqref{fun2} that $a=a'$.
\item $a,a' \in r_\si F^\si$. Consider $b,b'\in F^\si$ such that $a=r_\si b$ and $a'=r_\si b'$. Then $b'=r_\si w_\si^{\mathrm{aff}} r_\si b$ and from \eqref{fun2} it follows that $b=b'$ and thus $a=a'$.
\item \label{cannotoccur}$a'\in F$, $a \in r_\si F^\si$. From the equalities $a' = w_\si^\aff a=w_\si^\aff r_\si b$ with $b \in F^\si $ and from \eqref{fun2} it follows that
$a'=b$ and $w_\si^{\mathrm{aff}} r_\si \in \mathrm{Stab}_{W^{\mathrm{aff}}}(b)$. Then it holds that $\sigma \circ \psi \left(\mathrm{Stab}_{W^{\mathrm{aff}}}(b)\right)=\{\pm 1\},$ which would violate definition \eqref{Fsigma} and thus, this case cannot occur. 
\end{enumerate}
\item Firstly, consider $a\in F^\si$. Any element $w_\si^{\mathrm{aff}}\in W_\si^\aff $ stabilizing $a$ is also from $W^\aff$ and thus trivially $\mathrm{Stab}_{W^{\aff}_\si}(a) \subset \mathrm{Stab}_{W^{\aff}}(a)$. If, on the other hand, $w^{\aff}\in\mathrm{Stab}_{W^{\aff}}(a),$ then directly from definition \eqref{Fsigma} follows that $\si \circ \psi(w^\aff)=1$ and thus $w^{\aff}\in\mathrm{Stab}_{W_\si^{\aff}}(a)$, i.e. $\mathrm{Stab}_{W^{\aff}_\si}(a) = \mathrm{Stab}_{W^{\aff}}(a)$. Secondly, consider $a\in r_\si F^\si$. Then the stabilizer of a point $b\in F^\si $ such that $a= r_\si b$ is conjugated to the stabilizer of $a$, i.e. $\mathrm{Stab}_{W^{\aff}}(a) =r_\si \mathrm{Stab}_{W^{\aff}}(b)r_\si $ and one obtains $$r_\si \mathrm{Stab}_{W^{\aff}}(b)r_\si =r_\si \mathrm{Stab}_{W_\si^{\aff}}(b)r_\si = \mathrm{Stab}_{W_\si^{\aff}}(a).$$ 
If $a \in H^\si,$ then, according to \eqref{Hsigma}, the restricted homomorphism
$\si\circ\psi|_{\mathrm{Stab}_{W^{\aff}}(a)}$ maps $\mathrm{Stab}_{W^{\aff}}(a)$ surjectively onto $\{\pm 1\}$. The kernel of this homomorphism is
$\mathrm{Stab}_{W_\si^\aff}(\x)$ and thus the
isomorphism theorem gives
$$\mathrm{Stab}_{W^\aff}(\x) / \mathrm{Stab}_{W_\si^\aff}(\x) \cong \{\pm 1\}.$$
\end{enumerate}
\end{proof}

For any sign homomorphism $\si\in \Sigma$, the notion of the sets of boundaries $H^\si$ of the fundamental domain $F$ of $W^{\mathrm{aff}}$ is 
generalized to the affine even Weyl groups $W_{\sigma}^{\mathrm{aff}}$. The relation \eqref{Hsigma} is generalized by considering a second 
arbitrary sign homomorphism $\wt\si\in \Sigma$ and the subsets $H^{\wt\si,\si}$ of the boundaries of the fundamental domains 
$F\cup r_{\si}F^{\si}$ of $W_{\sigma}^{\mathrm{aff}}$ are introduced as  
\begin{equation}\label{Hsigmasigma}
H^{\wt{\si}, \si}= \set {\x \in F \cup r_{\si} F^{\si}}{\wt\si\circ\psi(\Stab_{W_{\si}^\aff}(a))=\{\pm 1 \}}.
\end{equation}
Generalizing relation \eqref{Fsigma2}, the set $F^{\wt\si,\si}$ is given as the fundamental domain $F \cup r_{\si} F^{\si}$ without its boundary points from $H^{\wt\si,\si}$,
\begin{equation}\label{Fsisi}
F^{\wt\si,\si} =(F \cup r_{\si} F^{\si})\setminus H^{\wt\si,\si}.
\end{equation}
It follows immediately from \eqref{Hsigmasigma} and \eqref{Fsisi} that, similarly to \eqref{Fsigma}, the set $F^{\wt\si,\si}$ can be expressed  as 
\begin{equation}\label{DOGhelp}
F^{\wt\si,\si}= \set {\x \in F\cup r_{\si}F^{\si}} {\wt\si\circ \psi(\Stab_{W_{\si}^{\aff}}(\x))=\{1 \}}.
\end{equation}
Note also that for the fundamental domain $F \cup r_{{\si}} F^{{\si}}$ it holds that
\begin{equation}\label{disdec}
F \cup r_{{\si}} F^{{\si}}=F^{\one,\si}.
\end{equation}
In order to determine the structure of the sets $H^{\wt{\si}, \si}$, their form is expressed using the simpler sets $H^{\si}$ in the following Proposition. 
\begin{tvr}\label{lemma123} For the sets $H^{\wt{\si},\si}$ it holds that
\begin{align}
H^{\wt{\si},\si}\cap F &= H^{\wt{\si}}\cap H^{\wt{\si} \cdot \si}. \label{HF}
\end{align}
\end{tvr}
\begin{proof}\
If $a\in F$ and $\x \in H^{\wt{\si},\si}$ then by \eqref{Hsigmasigma}
there exists an element $w^\aff \in \mathrm{Stab}_{W^\aff}(\x)$ 
such that 
\begin{align}
\si\circ \psi(w^\aff)&=1,\label{op1}\\
\wt{\si}\circ \psi(w^\aff)&= -1.\label{op2}
\end{align}
Then by \eqref{Hsigma} it holds that $a\in H^{\wt\si} $ and since also
\begin{equation}(\wt{\si}\cdot \si)\circ \psi(w^\aff) = [\wt{\si}\circ \psi(w^\aff)]\cdot[\si\circ \psi(w^\aff)]= -1,\label{op3}
\end{equation}
it follows that $a\in H^{\wt{\si} \cdot \si}$.

The statement $\x \in H^{\wt{\si}}\cap H^{\wt{\si} \cdot \si}$ is equivalent to existence of elements $w_1^\aff, w_2^\aff  \in \mathrm{Stab}_{W^\aff}(\x),$ such that 
\begin{align*}\wt{\si}\circ \psi(w_1^\aff) = -1, \\
[\wt{\si}\circ \psi(w_2^\aff)]\cdot[\si\circ \psi(w_2^\aff)]= -1. 
\end{align*}
If $\wt{\si}\circ \psi(w_2^\aff)=-1$,  $\si\circ \psi(w_2^\aff)=1$ or $\si\circ \psi(w_1^\aff) = 1,$ then equations \eqref{op1} -- \eqref{op3} are satisfied with $w^\aff =w_2^\aff$ or $w^\aff =w_1^\aff$, respectively. If $\wt{\si}\circ \psi(w_2^\aff)=1$,  $\si\circ \psi(w_2^\aff)=-1$ and $\si\circ \psi(w_1^\aff) = -1,$ then equations \eqref{op1} -- \eqref{op3} are satisfied with $w^\aff =w_1^\aff w_2^\aff$.
\end{proof}

Expressing the sets $F^{\wt{\si}, \si}$ in terms of the simpler sets $F^{\si}$, the
structure of $F^{\wt{\si}, \si}$ is uncovered in the following crucial theorem. 

\begin{thm}\label{thm123} The disjoint decomposition of the sets $F^{\wt{\si},\si}$ is the following
\begin{align}
F^{\wt\si,\si}= (F^{\wt\si}\cup  F^{\wt\si \cdot \si} ) \cup r_\si (F^{\wt\si}\cap  F^{\wt\si \cdot \si} ). \label{HF2}
\end{align}
\end{thm}
\begin{proof}\
Due to the disjoint decomposition \eqref{disdec}, firstly the set equality $F^{\wt\si,\si}\cap F =F^{\wt\si}\cup  F^{\wt\si \cdot \si}$ is derived using \eqref{HF} and \eqref{Fsigma2},  
\begin{align*}
F^{\wt\si,\si}\cap F= (F^{\one,\si}\setminus H^{\wt{\si},\si})\cap F= F\setminus ( H^{\wt{\si},\si}\cap F)=F\setminus (H^{\wt{\si}}\cap H^{\wt{\si} \cdot \si})= (F\setminus H^{\wt{\si}})\cup (F\setminus H^{\wt{\si} \cdot \si}).
\end{align*}

Secondly, the set equality $F^{\wt{\si},\si}\cap r_\si F^{\si}= r_\si(F^{\wt{\si}}\cap F^{\wt{\si}\cdot \si})$ is derived. Considering an element 
\begin{equation}\label{ab}
a=r_\si b\in F^{\wt{\si},\si}\cap r_\si F^\si, 	
\end{equation}
firstly $b \in F^{\wt{\si}}$ is shown.  Since it follows directly from \eqref{ab} that $b\in F^\si$, relation
\begin{equation}\label{opF1}
\{1\}=\wt\si\circ\psi(\Stab_{W_{\si}^\aff}(a)) 	=\wt\si\circ\psi(\Stab_{W_{\si}^\aff}(r_\si b))=\wt\si\circ\psi(\Stab_{W_{\si}^\aff}(b))
\end{equation}
and definition \eqref{DOGhelp} yield that $b\in F^{\wt{\si},\si}.$  Definition \eqref{Fsigma} implies for the point $b\in F^\si$ that for all $w^\aff \in \mathrm{Stab}_{W^\aff}(b)$ holds $\si\circ \psi(w^\aff)=1.$ Relation \eqref{calc1} grants that $\mathrm{Stab}_{W^\aff}(b)=\Stab_{W_{\si}^\aff}(b)$ and thus, applying $\wt{\si}\circ\psi$ on this equation and using relation \eqref{opF1}, it follows that 
\begin{equation}\label{opF2}
\wt\si\circ\psi(\Stab_{W^\aff}(b))=\wt\si\circ\psi(\Stab_{W_{\si}^\aff}(b))=\{1\},
\end{equation}
and $b\in F^{\wt\si} $ is proven. In order to prove $b \in F^{\wt{\si}\cdot\si}$, relations \eqref{Fsigma} and \eqref{opF2} guarantee for point $b \in F^\si$ and for all $w^\aff \in \mathrm{Stab}_{W^\aff}(b)$ that
\begin{equation}(\wt{\si}\cdot \si)\circ \psi(w^\aff) = [\wt{\si}\circ \psi(w^\aff)]\cdot[\si\circ \psi(w^\aff)]= 1.\label{opF3}
\end{equation}
and $b\in F^{\wt{\si}\cdot\si}$ is shown. Defining relation \eqref{ab} thus implies that $a \in r_\si (F^{\wt\si}\cap  F^{\wt\si \cdot \si} ).$ 

For the converse inclusion, let $a=r_\si b\in r_\si (F^{\wt\si}\cap  F^{\wt\si \cdot \si} ) $. Then $b\in F^{\wt\si}\cap  F^{\wt\si \cdot \si}$ implies for all $w^\aff \in \mathrm{Stab}_{W^\aff}(b)$ that $\wt \si\circ \psi(w^\aff)=1$ and together with \eqref{opF3} yields $ \si\circ \psi(w^\aff)=1$ and $b\in F^{\si}$. Then also, from \eqref{calc1},  $\mathrm{Stab}_{W^\aff}(b)=\Stab_{W_{\si}^\aff}(b)$ and \eqref{opF2} implies validity of \eqref{opF1}, i.e. $a \in F^{\wt{\si},\si}$.
\end{proof}

%%%%%%%%%%%%%%%%%%%%%%%%%%%%%%%%%%%%%%%%%%%%%%%%%%%%%%%%%%%%%%%%%%
%%%%%%%%%%%%%%%%%%%%%%%%%%%%%%%%%%%%%%%%%%%%%%%%%%%%%%%%%%%%%%%%%%
\subsection{Action of $W^\si$ on the maximal torus $\R^n/Q^{\vee}$}\setcounter{prop}{0}\
%%%%%%%%%%%%%%%%%%%%%%%%%%%%%%%%%%%%%%%%%%%%%%%%%%%%%%%%%%%%%%%%%%
%%%%%%%%%%%%%%%%%%%%%%%%%%%%%%%%%%%%%%%%%%%%%%%%%%%%%%%%%%%%%%%%%%

The isotropy subgroups of ${W_\si^{\aff}}$ and their orders are for any $a\in\R^n$ denoted by
\begin{equation*}
\mathrm{Stab}_{W_\si^{\aff}}(a) = \setb{w_\si^{\aff}\in W_\si^{\aff}}{{w_\si^{\aff}}a=a},\q h^\si(a)=|\mathrm{Stab}_{W_\si^{\aff}}(a)|.
\end{equation*}
Related functions $\ep^\si:\R^n\map\N$ are defined by the relation
\begin{equation}\label{epsi}
\ep^\si(a)=\frac{|W^\si|}{h^\si(a)}.
\end{equation}
Since for any $w_\si^{\aff}\in W_\si^{\aff}$ are the stabilizers $\mathrm{Stab}_{W_\si^{\aff}}(a) $ and $\mathrm{Stab}_{W_\si^{\aff}}(w^{\aff}a) $ conjugated, it holds that
\begin{equation}\label{epshift}
\ep^\si(a)=\ep^\si (w_\si^{\aff}a),\q w_\si^{\aff}\in W_\si^{\aff}.
\end{equation}

Considering the standard action of subgroups $W^\si$ on the torus $\R^n/Q^{\vee}$, the isotropy group of $x\in \R^n/Q^{\vee}$ and its order are denoted by
 \begin{equation}\label{hxe}
\wt{h}^\si(x)\equiv |\mathrm{Stab^\si} (x)|,\q\mathrm{Stab^\si} (x)=\set{w\in W^\si}{wx=x}.
\end{equation}
Denoting the orbit and its order by
 \begin{equation*}
\wt{\ep}^\si(x)\equiv |W^\si x|,\q W^\si x=\set{wx\in \R^n/Q^{\vee} }{w\in W^\si},
\end{equation*}
the orbit-stabilizer theorem gives that
\begin{equation*}%\label{epE}
\wt{\ep}^\si(x)=\frac{|W^\si|}{\wt{h}^\si(x)}.
\end{equation*}
The following proposition summarizes properties of the action of $W^\si$ on $\R^n/Q\dual$.

\begin{tvr}\label{torus}
\begin{enumerate}
\item For each $x\in \R^n/Q^{\vee},$ there exists $x'\in F^{\one,\si}/Q\dual$ and $w\in W^{\si}$ such that
\begin{equation}\label{rfun1E}
x=wx'.
\end{equation}
\item Consider $x,x' \in F^{\one,\si}/Q\dual$ and $w\in W^\si$  such that $wx=x'$, then
\begin{equation}\label{rfun2E}
x=x'=wx.
\end{equation}
\item Consider a point $x \in F^{\one,\si}/Q^{\vee}$ i.e. $x=\x+Q^{\vee},$ with $\x \in F^{\one,\si}.$ For the isotropy group it holds that
\begin{equation}\label{rfunstabE} 
\mathrm{Stab^\si}(x) \cong \mathrm{Stab}_{W^{\mathrm{aff}}_\si}(\x).
\end{equation}
\end{enumerate}
\end{tvr}
\begin{proof}\
\begin{enumerate}
\item Follows directly from (\ref{fun1s}).
\item Follows directly from (\ref{FD2}).
\item The isomorphism in relation \eqref{rfunstabE}, denoted by $\psi^\si_\x,$ is defined as the retraction homomorphism $\psi$ 
restricted to the isotropy group $\mathrm{Stab}_{{W}^{\mathrm{aff}}_\si}(\x)$, 
i.e.
$$\psi^\si_\x \equiv \psi|_{\mathrm{Stab}_{W^{\mathrm{aff}}_\si}(\x)}.$$
In order to prove surjectivity, note that assuming validity of $\x = w^{\aff}\x = w\x +q\dual$ implies $\x -w\x =q\dual \in Q\dual,$
i.e. it follows that $wx=x$ and vice versa. Considering injectivity, it holds for the kernel of $\psi^\si_\x$  that
$$\ker \psi^\si_\x =\ker \psi \cap \mathrm{Stab}_{W_\si^\aff}(\x),$$
and the kernel of $\psi$ consists of all possible $Q\dual-$translations. The only translation of $Q\dual$ that stabilizes $\x$ is the zero translation, therefore 
$\ker \psi^\si_a$ is trivial.
\end{enumerate}
\end{proof}

Equation \eqref{rfunstabE} implies the following straightforward relation between the orders of isotropy groups in $\R^n/Q\dual$ and 
in $\R^n$,
\begin{align}
\ep^\si(a)& = \wt{\ep}^\si(a+Q\dual),\label{epsilontilda} \\ \ h^\si(a)&=\wt{h}^\si(a+Q\dual).
\end{align}
Note that instead of  $\wt{h}^{\one}(x)$ and $\wt{\ep}^{\one}(x)$, the symbols $h_x$ and $\ep(x)$ are used in \cite{discliegrI} for $\abs{\Stab(x)}$ and $\abs{W x}$, respectively. The calculation
procedure of the coefficients $\wt{h}^{\one}(x)$ is detailed in \S 3.7 in \cite{discliegrI}.
Having calculated from this procedure and relation \eqref{epsilontilda} the values of  $h^{\one}(a)$, relation \eqref{calc1} allows to determine the remaining values $h^{\si}(a)$ for any $ \x \in F$
as
\begin{alignat}{1}\label{hcalc}
 h^\si(\x) &= \left\{ 
  \begin{array}{l l}
    \frac{1}{2}h^{\one}(\x) & \quad \text{if $\x \in H^{\si},$}\\
    h^{\one}(\x) & \quad \text{otherwise.}
  \end{array} \right. 
\end{alignat}
The last step is to extend the values of $h^{\si}(a),\, \x \in F$ to the entire fundamental domain $F^{\one,\si}$ of $W_{\sigma}^{\mathrm{aff}}$ via the following relation
$$h^\si (r_\si\x) = h^\si(\x). $$
Finally, the coefficients $\ep^\si(a),\,\x \in F^{\one,\si}$ are determined from $h^{\si}(a)$ by equation \eqref{epsi}.

%%%%%%%%%%%%%%%%%%%%%%%%%%%%%%%%%%%%%%%%%%%%%%%%%%%%%%%%%%%%%%%%%%
%%%%%%%%%%%%%%%%%%%%%%%%%%%%%%%%%%%%%%%%%%%%%%%%%%%%%%%%%%%%%%%%%%
\subsection{Dual affine Weyl group and its even subgroups}\setcounter{prop}{0}\
%%%%%%%%%%%%%%%%%%%%%%%%%%%%%%%%%%%%%%%%%%%%%%%%%%%%%%%%%%%%%%%%%%
%%%%%%%%%%%%%%%%%%%%%%%%%%%%%%%%%%%%%%%%%%%%%%%%%%%%%%%%%%%%%%%%%%

%The elements of the dual Cartan matrix $C^{\vee}$ are 
 %\begin{equation*}%\label{dCar}
% C^{\vee}_{ij}=\frac{2\sca{\al^{\vee}_i}{\al^{\vee}_j} }{\sca{\al^{\vee}_j}{\al^{\vee}_j}}=C_{ji},\q i,j\in\{1,\dots,n\}.
%\end{equation*}
%The positive weight lattice $P^{+}$ is defined as
%\begin{equation*}%\label{ddPP}
% P^{+}=\Z_{\geq 0} \om_1+\dots +\Z_{\geq 0} \om_n.
%\end{equation*}
The dual affine Weyl group $\widehat{W}^{\mathrm{aff}}$ is a semidirect product of the group of shifts from the root lattice $Q$ and the Weyl group $W$,
\begin{equation}\label{directd}
 \widehat{W}^{\mathrm{aff}}= Q \rtimes W.
\end{equation}
Therefore, for each $\wh w^\aff \in \wh W^\aff$ there exists a unique $w \in W$ and a unique $q \in Q$ such that for each $b \in \R^n$ it holds that
\begin{equation}\label{semiaff2}
\wh w^\aff b = w b +q.
\end{equation}
 The dual retraction homomorphism $\wh \psi : \wh W^{\aff} \map W$ of the semidirect product \eqref{directd} is given for the element of the form \eqref{semiaff2} as
\begin{equation}\label{retract2}
\wh \psi ( \wh w^\aff) = w.
\end{equation}

Equivalently, the dual affine Weyl group $\widehat{W}^{\mathrm{aff}}$ is generated by reflections $r_i$ and the reflection $r_0^{\vee}$ given by
\begin{equation*}%\label{daWeyl}
r_0^{\vee} \x=r_{\eta} \x + \frac{2\eta}{\sca{\eta}{\eta}}, \q r_{\eta}\x=\x-\frac{2\sca{\x}{\eta} }{\sca{\eta}{\eta}}\eta,\q \x\in\R^n.
\end{equation*}
The set of generators of  $\widehat{W}^{\mathrm{aff}}$ is denoted by $R\dual$, 
$$R\dual = \{r\dual_0,\,r_1,\,\dots,\,r_n\}.$$
Similarly to \eqref{deffun}, the fundamental region $F^\vee$ of $\widehat{W}^{\mathrm{aff}}$ is the convex hull of the vertices $\left\{ 0, \frac{\om_1}{m^{\vee}_1},\dots,\frac{\om_n}{m^{\vee}_n} \right\}$,
\begin{align}\label{defdfun}
F^\vee &=\setb{z_1\om_1+\dots+z_n\om_n}{z_0,\dots, z_n\in \R^{\geq 0}, \, z_0+z_1 m_1^{\vee}+\dots+z_n m^{\vee}_n=1} .
\end{align}

The corresponding dual affine even Weyl groups are the kernels of the expanded sign homomorphisms $\si \circ \wh{\psi}$ or, equivalently, the semidirect products of the group of translations $Q$ and the even Weyl groups $W^\si$,
\begin{equation*}
 \widehat{W}^{\mathrm{aff}}_\si=\setb {\widehat{w}^\aff \in \widehat{W}^\aff} { \sigma\circ \wh{\psi}(\widehat{w}^\aff)=1}= Q \rtimes W^\si.
\end{equation*}
The set of generators of the affine Weyl group $\widehat{W}^\aff$ with negative values of the sign homomorphisms $\si\circ \wh{\psi}$ is denoted by $R^{\vee\si}$,
\begin{equation*}
 R^{\vee\sigma} = \set{r\in R\dual}{\si\circ\wh{\psi}(r)=-1}.
\end{equation*}
Analogously to \eqref{Hsigma}, the set $H^{\vee\sigma}$ is defined as 
\begin{alignat*}{1}
H^{\vee\sigma} &=\set{b \in F\dual}{\si\circ\wh{\psi}(\Stab_{\wh{W}^{\aff}}(b))=\{\pm 1\}}.
\end{alignat*}
and similarly to \eqref{Fsigma} and \eqref{Fsigma2}, the following holds for the domain $F^{\vee\sigma}$,
 \begin{alignat*}{1}
F^{\vee\sigma} &= F\dual\setminus H^{\vee\si} =\set{b \in F\dual}{\si\circ\wh{\psi}(\Stab_{\wh{W}^{\aff}}(b))= \{1\}}.
\end{alignat*}
The fundamental domains of the dual even affine Weyl groups $\wh{W}_{\si}^{\aff}$ are determined analogously to as in Proposition \ref{tvrfund}.
\begin{tvr}\label{tvrfunddual}The set $F^\vee \cup r_\si F^{\vee\si}$ is for any $r^\vee_\si \in R^{\vee\si}$ a fundamental domain of $\wh{W}_\si^{\aff}$ satisfying the following conditions.
\begin{enumerate}
\item For any $b\in \R^n,$ there exist $b' \in F\dual\cup r^\vee_\si F^{\vee\si}$ and $\wh{w}_{\si}^{\aff} \in \wh{W}_{\si}^{\aff}$ such that
\begin{equation}\label{fun1Ed}
\wh{w}_{\si}^{\aff}b =b'.
\end{equation}
\item If $b,b' \in F\dual\cup r^\vee_\si F^{\vee\si}$ and $\wh{w}_{\si}^{\aff}b=b',$ $\wh{w}_{\si}^{\aff} \in \wh{W}_{\si}^{\aff},$ then
\begin{equation}\label{fun2Ed}
b=b'=\wh{w}_{\si}^{\aff}b.
\end{equation}
\item Consider a point $b \in F\dual \cup r^\vee_\si F^{\vee\si}.$ If $b \in F^{\vee\si}$ or $b \in r^\vee_\si F^{\vee\si},$ then it holds that 
\begin{equation}\label{fun3Ed}
\mathrm{Stab}_{\widehat{W}^{\mathrm{aff}}_\si}(b) =\mathrm{Stab}_{\widehat{W}^{\mathrm{aff}}}(b),
\end{equation}
if $b \in H^{\vee\si},$ then it holds that
$$\abs{\mathrm{Stab}_{\widehat{W}^{\mathrm{aff}}}(b)} = 2\abs{\mathrm{Stab}_{\widehat{W}^{\mathrm{aff}}_\si}(b)}.$$
\end{enumerate}
\end{tvr}
The subsets $H^{\vee\wt\si,{\si}}$ of the boundaries of the fundamental domains 
$F\dual\cup r^\vee_{{\si}}F^{\vee{\si}}$ are given as
\begin{equation}\label{Hsigmasigmadual}
H^{\vee\wt\si,{\si}} = \set{b\in F\dual\cup r^\vee_{{\si}}F^{\vee{\si}} }{\wt\si\circ \wh{\psi} (\Stab_{\wh{W}^{\aff}_{{\si}}}(b))=\{\pm 1\}}
\end{equation}
and the dual analogue of  $F^{\wt\si,{\si}}$ is given as
\begin{equation}\label{dualweightdomain}
F^{\vee\wt\si,{\si}}=(F^{\vee} \cup r^\vee_{{\si}} F^{\vee{\si}})\setminus H^{\vee\wt\si,{\si}}=\set{b \in F\dual\cup r^\vee_{{\si}} F^{\vee{\si}} }{\wt\si \circ \wh{\psi} (\Stab_{\wh{W}_{{\si}}^{\aff}}(b))=\{1 \}}.
\end{equation}
The structure of the sets $F^{\vee\wt{\si}, \si}$ and $H^{\vee\wt{\si}, \si}$ is expressed using the simpler sets $H^{\vee\si}$, $F^{\vee\si}$ as in Proposition \ref{lemma123} and Theorem \ref{thm123}. 
\begin{thm}\label{thm123d} For the sets $F^{\vee\wt{\si},\si}$ and $H^{\vee\wt{\si},\si}$ it holds that
\begin{align}
F^{\vee\wt\si,\si}&= (F^{\vee\wt\si}\cup  F^{\vee\wt\si \cdot \si} ) \cup r^\vee_\si (F^{\vee\wt\si}\cap  F^{\vee\wt\si \cdot \si} ),\label{HFd2}\\
H^{\vee\wt{\si},\si}\cap F^\vee &= H^{\vee\wt{\si}}\cap H^{\vee\wt{\si} \cdot \si}.\label{HFd}
\end{align}
\end{thm}

\subsection{Action of $W^\si$ on the torus $\R^n/MQ$}\setcounter{prop}{0}\

The isotropy subgroups of ${\widehat W_\si^{\aff}}$ are for any $b\in\R^n$ denoted by
\begin{equation*}
\mathrm{Stab}_{\widehat W_\si^{\aff}}(b) = \setb{w_\si^{\aff}\in \widehat W_\si^{\aff}}{{w_\si^{\aff}}b=b}
\end{equation*}
and the orders of the stabilizers $\mathrm{Stab}_{\widehat W_\si^{\aff}}(b/M),\,M\in \N$ are denoted by
\begin{equation}\label{hM}
h^{\vee\si}_M(b) = \abs{\mathrm{Stab}_{\widehat W_\si^{\aff}}\left(\frac{b}{M}\right)} .
\end{equation}
Invariance of the root lattice $Q$ with respect to the Weyl group $W$ induces invariance of $Q$ with respect to all its even subgroups $W^\si$. This implies that the action of any $W^\si$ on the quotient group $\R^n/MQ$ is well-defined for any $M\in \N$.  
Considering this action of subgroups $W^\si$ on the torus $\R^n/MQ$, the isotropy group of $\la \in \R^n/MQ$ and its order are denoted by 
\begin{equation}\label{hla}
\wt{h}^{\vee\si}_M (\la)\equiv |\mathrm{Stab}^{\vee\si} (\la)|,\q \mathrm{Stab}^{\vee\si} (\la)=\set{w\in W^\si}{w\la=\la}.
\end{equation}
The following proposition generalizes Proposition \ref{torus}, property (3)  to the action of $W^\si$ on $\R^n/MQ$.
\begin{tvr}
Consider a point $\la\in MF^{\vee\one,\si} \cap \R^n/MQ$, i.e. $\la=b+MQ$, with $b \in MF^{\vee\one,\si}.$ For the isotropy group it holds that
\begin{equation}\label{rfunstab2}
\mathrm{Stab}^{\vee\si} (\la) \cong \mathrm{Stab}_{\widehat{W}^{\mathrm{aff}}_\si}\left(\frac{b}{M}\right).
\end{equation}
\end{tvr}
\begin{proof}\
The isomorphism in relation \eqref{rfunstab2}, denoted by $\wh{\psi}^\si_b,$ is defined as the retraction homomorphism $\wh{\psi}$ 
restricted to the isotropy group $\mathrm{Stab}_{\wh{W}^{\mathrm{aff}}_\si}(b/M)$, 
i.e.
$$\wh{\psi}^\si_b \equiv \wh{\psi}|_{\mathrm{Stab}_{\wh{W}^{\mathrm{aff}}_\si}(\frac{b}{M})}.$$
In order to prove surjectivity, note that assuming validity of $b/M =w^{\aff}(b/M)= w(b/M) +q$ implies $b -wb =Mq \in MQ,$
i.e. it follows that $w\la=\la$ and vice versa. Considering injectivity, it holds for the kernel of $\wh{\psi}^\si_b$ that
$$\ker \wh{\psi}^\si_b =\ker \wh{\psi} \cap \mathrm{Stab}_{\widehat{W}_\si^\aff}\left(\frac{b}{M}\right).$$
and the kernel of $\wh{\psi}$ consists of all possible $Q-$translations. The only translation of $Q$ that stabilizes $b/M$ is the zero translation, therefore 
$\ker \wh{\psi}^\si_b$ is trivial.
\end{proof}

Equation \eqref{rfunstab2} implies the following straightforward relation between the orders of isotropy groups in $P/MQ$ and 
$P$,
\begin{equation}\label{epsilontilda2}
h_M^{\vee\si}(b)=\wt{h}_M^{\vee\si}(b+MQ). 
\end{equation}
Note that instead of $\wt{h}_M^{\vee\one}(\la)$, the symbol $h_{\la}^{\vee}$ is used in \cite{discliegrII} for $\abs{\Stab^{\vee}(\la)}$.
The calculation procedure of the coefficients $\wt{h}_M^{\vee\one}(x)$ is detailed in \S 3.7 in \cite{discliegrI}.
Having calculated from this procedure and relation \eqref{epsilontilda2} the values of $h_M^{\vee\one}(b)$, relation \eqref{fun3Ed} allows to determine the remaining values
$h^{\vee \si}_M(b)$ for any $b\in MF^\vee$ as
\begin{alignat*}{1}
h^{\vee \si}_M(b) &= \left\{ 
  \begin{array}{l l}
    \frac{1}{2}h^{\vee\one }_M(b) & \quad \text{if $b/M \in H^{\vee\si},$}\\
    h^{\vee\one }_M(b) & \quad \text{otherwise.}
  \end{array} \right.
\end{alignat*} 
The last step is to extend the values of $h^{\vee\si}_M(b),\, b\in MF^\vee$ to the entire magnified fundamental domain $MF^{\vee\one,\si}$ of $\widehat W_{\sigma}^{\mathrm{aff}}$ via the following relation
$$h^{\vee \si}_M (r_\si b) = h^{\vee \si}_M(b). $$

%%%%%%%%%%%%%%%%%%%%%%%%%%%%%%%%%%%%%%%%%%%%%%%%%%%%%%%%%%%%%%%%%%
%%%%%%%%%%%%%%%%%%%%%%%%%%%%%%%%%%%%%%%%%%%%%%%%%%%%%%%%%%%%%%%%%%
\section{Orbit functions}\setcounter{prop}{0}
%%%%%%%%%%%%%%%%%%%%%%%%%%%%%%%%%%%%%%%%%%%%%%%%%%%%%%%%%%%%%%%%%%
%%%%%%%%%%%%%%%%%%%%%%%%%%%%%%%%%%%%%%%%%%%%%%%%%%%%%%%%%%%%%%%%%%

\subsection{Ten types of orbit functions}\setcounter{prop}{0}\

Consider a sign homomorphism $\si \in \Sigma$ and the corresponding even subgroup $W^\si \subset W$. Taking another sign homomorphism $\wt\si\in \Sigma$ and a parameter $b \in \R^n$, the most general form of orbit functions $\Psi_b^{\wt{\si},{\si}}: \R^n\map \Com$ is introduced as
\begin{equation}\label{genf}
\Psi_b^{\wt{\si},{\si}}(\x)= \sum_{w \in W^{\si}} {\wt{\si}(w)\e^{2\pi \im \scalar{wb}{\x}}}.
\end{equation}
This general definition leads to three types of orbit functions for root systems with one root-length and to ten types of orbit functions for root systems with two root-lengths. 

If $\si=\wt\si = \one,$ then the general form of orbit functions \eqref{genf} specializes to the standard symmetric orbit sums. The symmetric orbit sums, studied already in \cite{Bour}, are called $C-$functions and reviewed extensively in \cite{KP1}. These orbit functions, summed over to the entire Weyl group, are denoted by the symbol $\Phi$ in \cite{discliegrI,discliegrII}, 
$$\Phi\equiv\Psi^{\one,\one}.$$
If $\si=\one$ and $\wt\si = \si^e,$ then the general form of orbit functions \eqref{genf} specializes to the standard antisymmetric orbit sums. The antisymmetric orbit sums, studied also in \cite{Bour}, are called $S-$functions and reviewed extensively in \cite{KP2}. These most ubiquitous orbit functions, which also appear in the Weyl character formula, are denoted by the symbol $\varphi$ in \cite{discliegrI,discliegrII}, 
$$\varphi\equiv\Psi^{\si^e,\one}.$$
Both choices $\si=\si^e,\,\wt\si = \one$ and $\si=\si^e,\,\wt\si = \si^e$ in \eqref{genf} lead to the same type of functions. These orbit functions, called $E-$functions, are sums with respect to the standard even Weyl group. They are reviewed in \cite{KP3}, denoted by $\Xi$ in \cite{discliegrE} and by $\Xi^{e+}$ in \cite{Ef2d},
\begin{alignat*}{1}
\Xi^{e+}&\equiv \Psi^{\one,\si^e}=\Psi^{\si^e,\si^e}.
\end{alignat*}
The standard $C-$, $S-$, and $E-$functions above exist for all root systems. 

For root systems with two different lengths of roots there exist two more types of $S-$functions called $S^s-$ and $S^l-$functions in \cite{SsSlcub} and denoted by $\varphi^{s}$ and $\varphi^{l}$ in \cite{discliegrII},     
\begin{alignat*}{1}
\varphi^{s}&\equiv \Psi^{\si^s,\one},\\
\varphi^{l}&\equiv \Psi^{\si^l,\one}.
\end{alignat*}
The following combinations of $\si$ and $\wt\si$ lead to two more types of $E-$functions denoted by $\Xi^{s+}$ and $\Xi^{l+}$ in \cite{Ef2d},
\begin{alignat*}{1}
\Xi^{s+}&\equiv \Psi^{\one,\si^s}=\Psi^{\si^s,\si^s},\\
\Xi^{l+}&\equiv \Psi^{\one,\si^l}=\Psi^{\si^l,\si^l},
\end{alignat*}
and lastly the remaining choices $\wt{\si}\neq \si$ and $\si \neq \one$ lead to the three types of functions called mixed $E-$functions and are denoted by $\Xi^{e-}$, $\Xi^{s-}$ and $\Xi^{l-}$ in \cite{Ef2d},
\begin{alignat*}{1}
\Xi^{e-}&\equiv \Psi^{\si^l,\si^e}=\Psi^{\si^s,\si^e},\\
\Xi^{s-}&\equiv \Psi^{\si^l,\si^s}=\Psi^{\si^e,\si^s},\\
\Xi^{l-}&\equiv \Psi^{\si^s,\si^l}=\Psi^{\si^e,\si^l}.
\end{alignat*}

\subsection{Argument and label symmetries}\setcounter{prop}{0}\

First, the argument symmetry of the general form of orbit functions \eqref{genf} is investigated  by restricting their label to the weight lattice $b \in P.$ It appears that the composition of the sign homomorphism $\wt\si \in \Sigma$ and the retraction homomorphism $\psi$, given by \eqref{retract}, determines the sign change of the argument invariance. The following proposition describes the argument invariance properties of all types of orbit functions. 
\begin{tvr}\label{inv}
 Let $b \in P,$ then for any $w^{\aff} \in W_{\si}^{\aff}$ and any $\x \in \R^n$ it holds that
\begin{equation}\label{xinvariant}
\Psi_b^{\wt{\si},\si}(w^{\aff}\x)=\wt{\si}\circ\psi(w^{\aff})\cdot\Psi_b^{\wt{\si},\si}(\x).
\end{equation}
Additionally, the values of all orbit functions $\Psi_b^{\wt{\si},\si}$ are zero for $\x'\in H^{\wt{\si},\si}$, 
\begin{equation}\label{zerop}
\Psi_b^{\wt{\si},\si}(\x ')=0,\ \x ' \in H^{\wt{\si},\si}. 
\end{equation}
\end{tvr}
\begin{proof}\
Consider an element of the even affine Weyl group $W_{\si}^{\aff}$ of the form $w^{\aff}\x=w'a+q\dual,$ with
$w' \in W^{{\si}}$ and $q\dual\in Q\dual$. Since the lattice $P$ is $W-$invariant it holds that  $wb \in P$ for any $w\in W$. Thus $\scalar{wb}{q\dual}\in \Z$ and it follows that 
$$\Psi_b^{\wt{\si},\si}(w^{\aff}\x)= \sum_{w \in W^{\si}} {\wt{\si}(w)\e^{2\pi \im \scalar{wb}{w'\x+q\dual}}}=\sum_{w \in W^{\si}} {\wt{\si}(w)\e^{2\pi \im \scalar{wb}{w'\x}}}.$$
Rearranging the terms yields
$$\sum_{w \in W^{\si}} {\wt{\si}(w)\e^{2\pi \im \scalar{wb}{w'\x}}}=\wt{\si}(w')\cdot \sum_{w \in W^{\si}} {\wt{\si}(w)\e^{2\pi \im \scalar{wb}{\x}}}=\wt{\si}\circ\psi(w^\aff)\cdot\Psi_b^{\wt{\si},\si}(\x).$$
If $\x '\in H^{\wt{\si},\si},$ then by \eqref{Hsigmasigma} there exists an element $w^{\aff} \in \Stab_{W_{\si}^{\aff}}(\x ')$ such that $\wt{\si}\circ\psi(w^{\aff})= -1$. Therefore it follows that
$$\Psi_b^{\wt{\si},\si}(\x ')=\Psi_b^{\wt{\si},\si}(w^{\aff}\x ')=\wt{\si}\circ\psi(w^{\aff})\cdot\Psi_b^{\wt{\si},\si}(\x ')= -\Psi_b^{\wt{\si},\si}(\x ').$$
\end{proof}

Secondly, the label symmetry of the general form of orbit functions \eqref{genf} is induced by restricting their arguments to the refined coweight lattice $a \in \frac{1}{M}P^{\vee}.$ It appears that the composition of the sign homomorphism $\wt\si \in \Sigma$ and the dual retraction homomorphism $\wh \psi$, given by \eqref{retract2}, determines the sign change of the label invariance. The following proposition summarizes the label invariance properties of all types of orbit functions in a compact form.
\begin{tvr}\label{labelthm}
Let $\x \in \frac{1}{M}P^{\vee},$ then for any $\wh{w}^{\aff}\in \wh{W}_{\si}^{\aff}$ and any $b \in \R^n$ it holds that
$$\Psi_{M\wh{w}^{\aff}\left(\frac{b}{M}\right)}^{\wt{\si},\si}(\x)= \wt{\si}\circ\wh{\psi}(\wh{w}^{\aff})\cdot\Psi_{b}^{\wt{\si},\si}(\x).$$
Additionally,  the orbit function $\Psi_{b'}^{\wt{\si},\si}$ vanishes for any $b ' \in MH^{\vee\wt{\si},\si}$,
$$\Psi_{b'}^{\wt{\si},\si}(a)=0,\ b' \in MH^{\vee\wt{\si},\si}. $$
\end{tvr}
\begin{proof}
Consider an element of the dual affine even Weyl group of the form $\wh{w}^{\aff} b=w'b+q$, with
$w' \in W^{\si}$ and $q\in Q$.  Since the lattice $Q$ is $W-$invariant it holds that  $wq \in Q$ for any $w\in W$. Considering that $Ma \in P^{\vee},$ it holds that $\scalar{Mwq}{a}=\scalar{wq}{Ma}\in \Z$ and thus
$$\Psi_{M\wh{w}^{\aff}\left(\frac{b}{M}\right)}^{\wt{\si},\si}(\x)= \sum_{w \in W^{\si}} {\wt{\si}(w)\e^{2\pi \im \scalar{ww'b+Mwq}{\x}}}=\sum_{w \in W^{\si}} {\wt{\si}(w)\e^{2\pi \im \scalar{ww'b}{\x}}}.$$
Rearranging the terms yields
$$\sum_{w \in W^{\si}} {\wt{\si}(w)\e^{2\pi \im \scalar{ww'b}{\x}}}=\wt{\si}(w')\cdot\sum_{w \in W^{\si}} {\wt{\si}(w)\e^{2\pi \im \scalar{wb}{\x}}}=\wt{\si}\circ\wh{\psi}(\wh{w}^\aff)\cdot\Psi_b^{\wt{\si},\si}(\x).$$
If $b'\in MH^{\vee\wt{\si},\si},$ then by \eqref{Hsigmasigmadual} there exists an element $\wh{w}^{\aff} \in \Stab_{\wh{W}_{\si}^{\aff}}(b'/M)$ such that $\wt{\si}\circ\wh{\psi}(\wh{w}^{\aff})= -1.$ 
Therefore it follows that
$$\Psi_{b'}^{\wt{\si},\si}(a)=\Psi_{M\wh{w}^{\aff}\left(\frac{b'}{M}\right)}^{\wt{\si},\si}(a)=\wt{\si}\circ\wh{\psi}(\wh{w}^{\aff})\cdot\Psi_{b'}^{\wt{\si},\si}(a)= -\Psi_{b'}^{\wt{\si},\si}(a).$$
\end{proof}

%%%%%%%%%%%%%%%%%%%%%%%%%%%%%%%%%%%%%%%
\subsection{Hartley orbit functions}\setcounter{prop}{0}\
%%%%%%%%%%%%%%%%%%%%%%%%%%%%%%%%%%%%%%%%%%%%%

An important real-valued modification of orbit functions which for $a\in \R$ uses the Hartley kernel $$\cas (\x)= \cos(a)+\sin(a)$$ instead of exponential in \eqref{genf} is introduced. 
Fixing an even subgroup $W^\si \subset W$, an additional sign homomorphism $\wt\si\in \Sigma$ and a parameter $b \in \R^n$, the Hartley orbit functions $\Hart_b^{\wt{\si},{\si}}: \R^n\map \R$ are defined via relation
\begin{equation}\label{Hartdef}
\Hart_b^{\wt{\si},\si}(a)= \sum_{w \in W^{\si}} {\wt{\si}(w)\, \cas (2\pi  \scalar{wb}{\x})}.
\end{equation}
Similarly to \eqref{genf}, such definition leads to three types of real-valued orbit functions for root systems with one root-length and to ten types of orbit functions for root systems with two root-lengths. Note that the relation of exponential function to the $\cas$ function 
implies
\begin{equation}\label{HartS}
\Hart_b^{\wt{\si},\si} = \Rere \Psi_b^{\wt{\si},\si}+\Imim \Psi_b^{\wt{\si},\si}.
\end{equation}
This property immediately allows to replicate the argument-label symmetries formulated in Propositions \ref{inv} and \ref{labelthm} as follows.
\begin{tvr}\label{invH}
 Let $b \in P,$ then for any $w^{\aff} \in W_{\si}^{\aff}$ and any $\x \in \R^n$ it holds that
\begin{equation}\label{xinvariantH}
\Hart_b^{\wt{\si},\si}(w^{\aff}\x)=\wt{\si}\circ\psi(w^{\aff})\cdot\Hart_b^{\wt{\si},\si}(\x).
\end{equation}
Additionally, the values of all Hartley orbit functions $\Hart_b^{\wt{\si},\si}$ are zero for $\x'\in H^{\wt{\si},\si}$, 
\begin{equation}\label{zeropH}
\Hart_b^{\wt{\si},\si}(\x ')=0,\ \x ' \in H^{\wt{\si},\si}. 
\end{equation}
\end{tvr}
\begin{tvr}\label{labelthmH}
Let $\x \in \frac{1}{M}P^{\vee},$ then for any $\wh{w}^{\aff}\in \wh{W}_{\si}^{\aff}$ and any $b \in \R^n$ it holds that
$$\Hart_{M\wh{w}^{\aff}\left(\frac{b}{M}\right)}^{\wt{\si},\si}(\x)= \wt{\si}\circ\wh{\psi}(\wh{w}^{\aff})\cdot\Hart_{b}^{\wt{\si},\si}(\x).$$
Additionally,  the Hartley orbit function $\Hart_{b'}^{\wt{\si},\si}$ vanishes for any $b ' \in MH^{\vee\wt{\si},\si}$,
$$\Hart_{b'}^{\wt{\si},\si}(a)=0,\ b' \in MH^{\vee\wt{\si},\si}. $$
\end{tvr}
\begin{example}\label{ex0}
Consider the Lie algebra $G_2$ and the corresponding Hartley orbit functions. Since the standard orbit functions $\Psi^{\one,\one}$, $\Psi^{\si^e ,\one}$ and $\Psi^{\one, \si^e }$ of $G_2$ are real-valued \cite{KP1,KP2,KP3} and $\Psi^{\si^s,\one,}$, $\Psi^{\si^l,\one}$ and $\Psi^{\si^s, \si^e }$ are purely imaginary \cite{Ef2d,SsSlcub}, it is sufficient to detail the four remaining functions  $\Hart^{\one,\si^s}$, $\Hart^{\one,\si^l}$ and $\Hart^{\si^e,\si^s}$, $\Hart^{\si^e,\si^l}$. Explicit expressions of these four types of functions are directly derived from explicit forms in \cite{Ef2d}. For instance, for a point with coordinates in $\alpha^\vee$-basis $(x,y)= x\al_1^\vee+y\al_2^\vee$ and a weight in $\om$-basis $(a,b)=a\om_1+b\om_2$ it holds that 
\begin{align*}
\Hart^{\si^e,\si^s}_{(a,b)}(x,y)=& \cas(2\pi (ax+by))- \cas(2\pi (-ax+(3a+b)y))-  \cas(2\pi ((2a+b)x-(3a+2b)y))\\ &+  \cas(2\pi ((a+b)x-(3a+2b)y))+  \cas(2\pi (-(2a+b)x+(3a+b)y))\\ &-  \cas(2\pi (-(a+b)x+by)).
\end{align*}
Explicit expressions of the other three types can be obtained similarly. The contour plots of the lowest functions $\Hart^{\one,\si^s}$, $\Hart^{\one,\si^l}$ and $\Hart^{\si^e,\si^s}$, $\Hart^{\si^e,\si^l}$ are depicted in Figures \ref{HEsG2}, \ref{HElG2},  \ref{HEsmG2} and \ref{HElmG2}.
\begin{figure}[t]
\resizebox{3.0cm}{!}{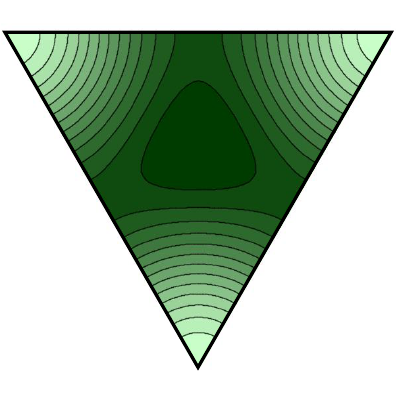}\hspace{6mm}\resizebox{3.0cm}{!}{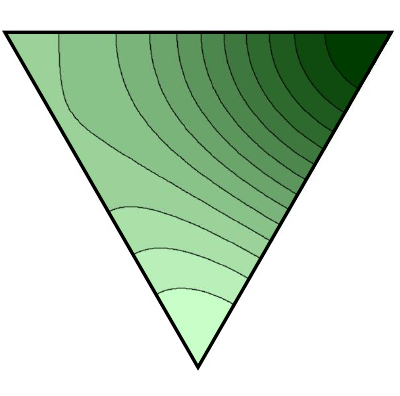}\hspace{6mm}\resizebox{3.0cm}{!}{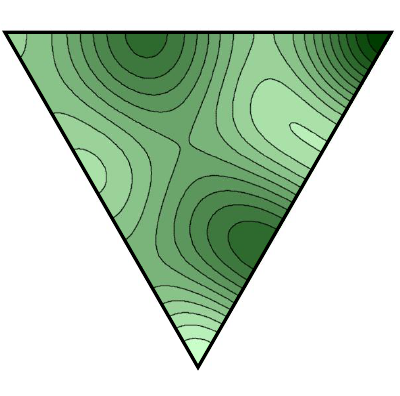}\hspace{6mm}\resizebox{3.0cm}{!}{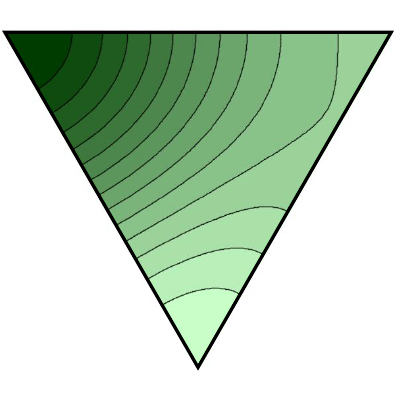}\\ \vspace{3mm}
\resizebox{3.0cm}{!}{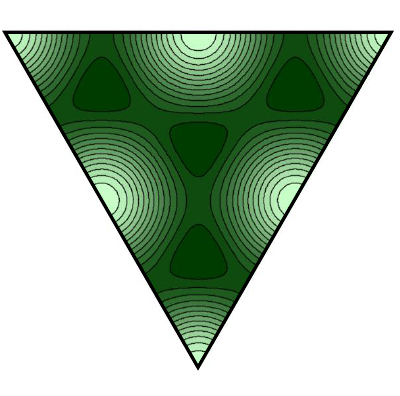}\hspace{6mm}\resizebox{3.0cm}{!}{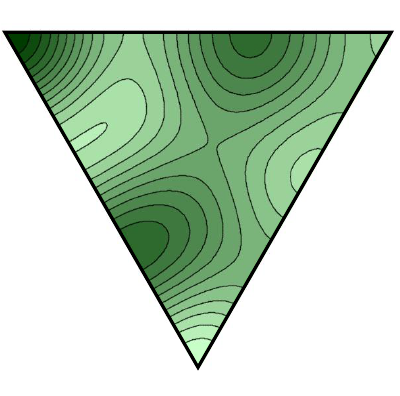}\hspace{6mm}\resizebox{3.0cm}{!}{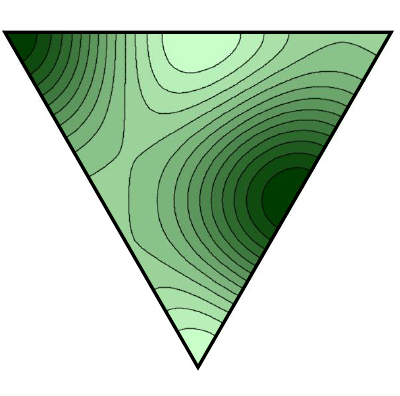}\hspace{6mm}\resizebox{3.0cm}{!}{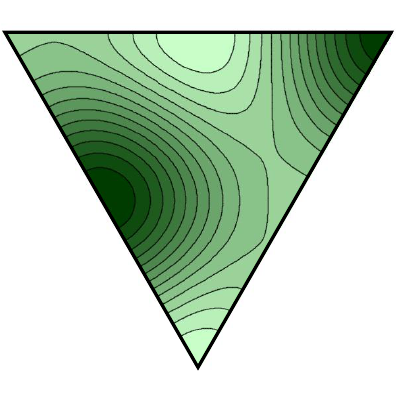}
\caption{\small The contour plots of Hartley orbit functions $\Hart^{\one,\si^s}$  over the fundamental domain $F^{\one,\si^s}$ of $G_2$.}\label{HEsG2}
\end{figure}
\begin{figure}[t]
\resizebox{3.0cm}{!}{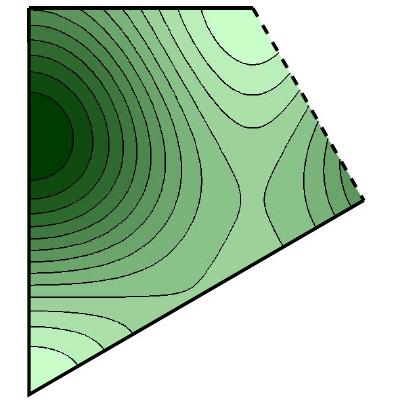}\hspace{6mm}\resizebox{3.0cm}{!}{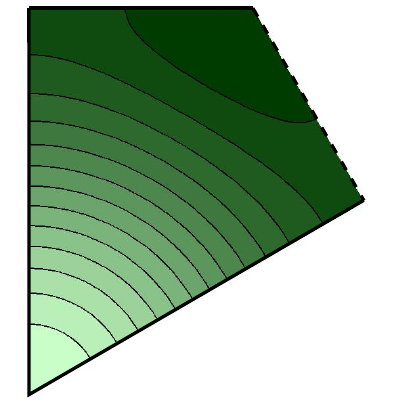}\hspace{6mm}\resizebox{3.0cm}{!}{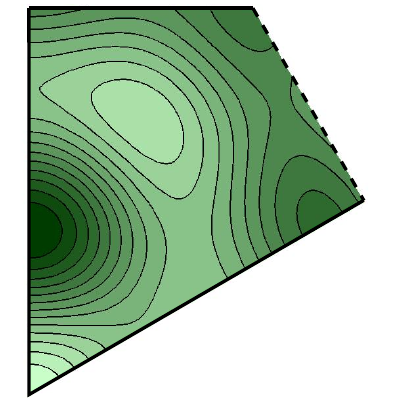}\hspace{6mm}\resizebox{3.0cm}{!}{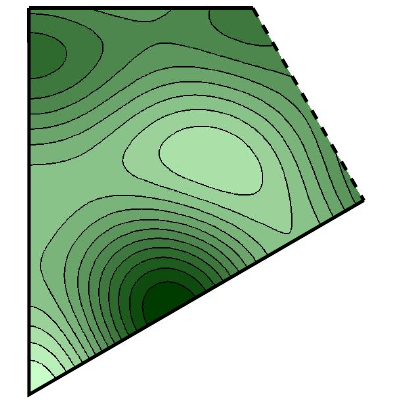}\\ \vspace{3mm}
\resizebox{3.0cm}{!}{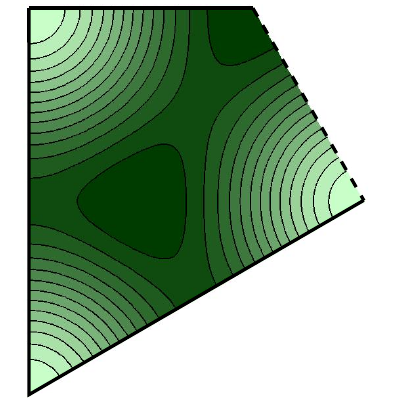}\hspace{6mm}\resizebox{3.0cm}{!}{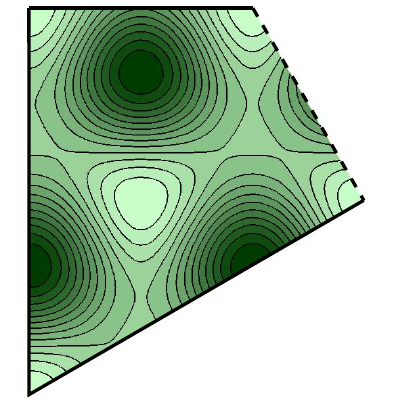}\hspace{6mm}\resizebox{3.0cm}{!}{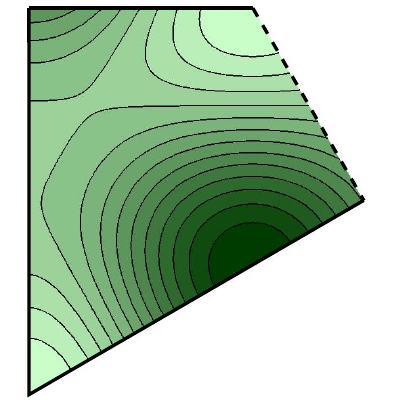}\hspace{6mm}\resizebox{3.0cm}{!}{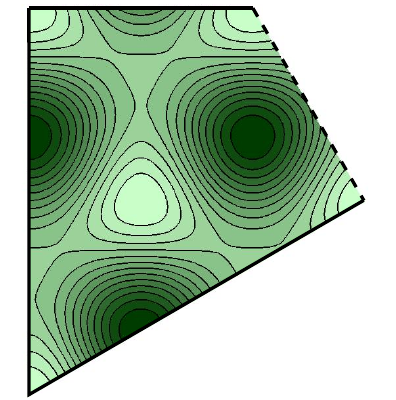}
\caption{\small The contour plots of Hartley orbit functions $\Hart^{\one,\si^l}$ over the fundamental domain $F^{\one,\si^l}$ of $G_2$. The dashed boundary does not belong to the fundamental domain.}\label{HElG2}
\end{figure}
\begin{figure}[t]
\resizebox{3.0cm}{!}{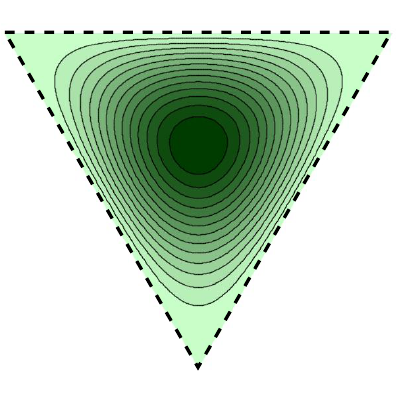}\hspace{6mm}\resizebox{3.0cm}{!}{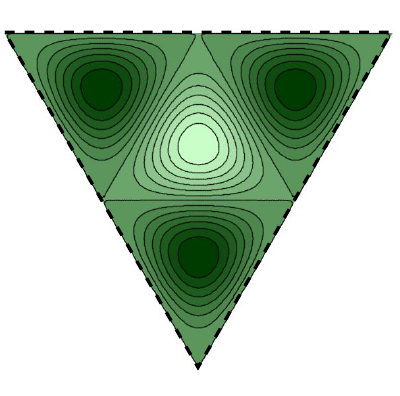}\hspace{6mm}\resizebox{3.0cm}{!}{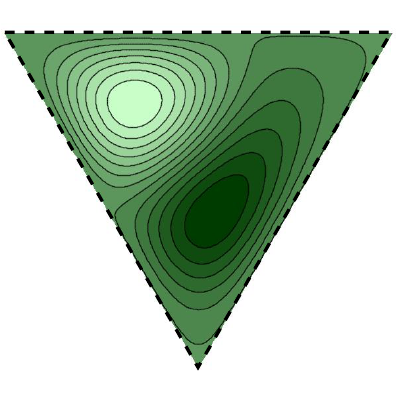}\hspace{6mm}\resizebox{3.0cm}{!}{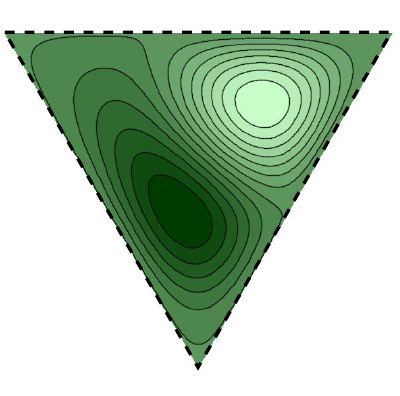}\\ \vspace{3mm}
\resizebox{3.0cm}{!}{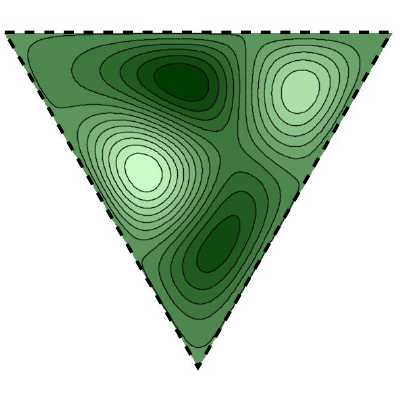}\hspace{6mm}\resizebox{3.0cm}{!}{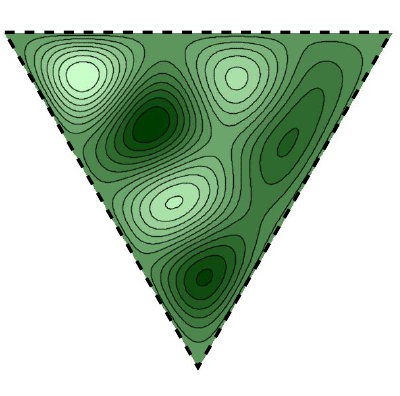}\hspace{6mm}\resizebox{3.0cm}{!}{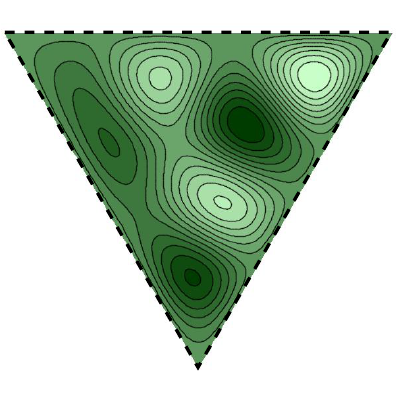}\hspace{6mm}\resizebox{3.0cm}{!}{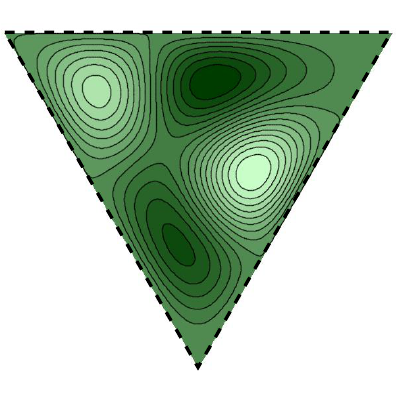}
\caption{\small The contour plots of Hartley orbit functions $\Hart^{\si^e,\si^s}$ over the fundamental domain $F^{\si^e,\si^s}$ of $G_2$. The dashed boundary does not belong to the fundamental domain.}\label{HEsmG2}
\end{figure}

\begin{figure}[t]
\resizebox{3.0cm}{!}{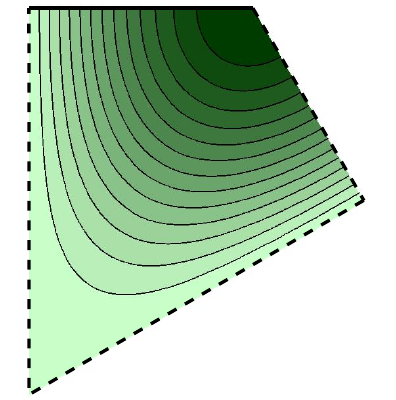}\hspace{6mm}\resizebox{3.0cm}{!}{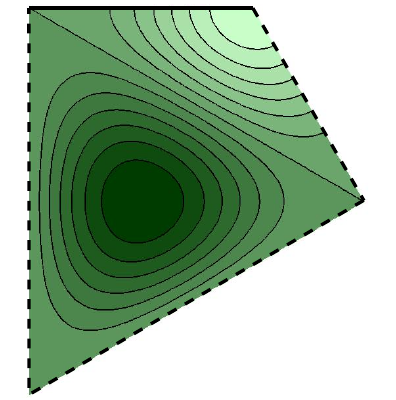}\hspace{6mm}\resizebox{3.0cm}{!}{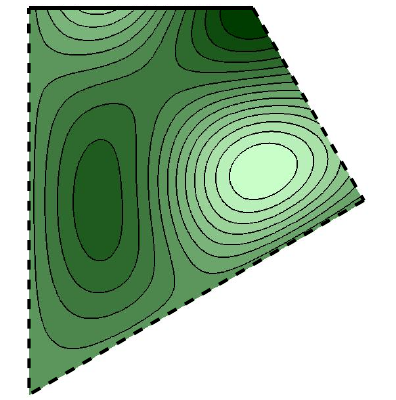}\hspace{6mm}\resizebox{3.0cm}{!}{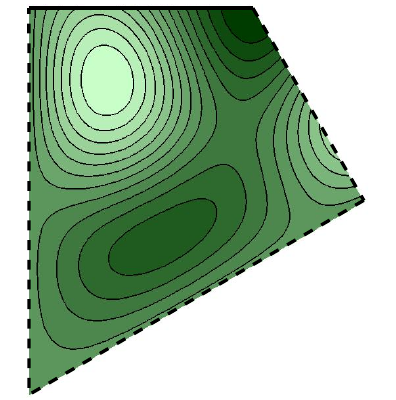}\\ \vspace{3mm}
\resizebox{3.0cm}{!}{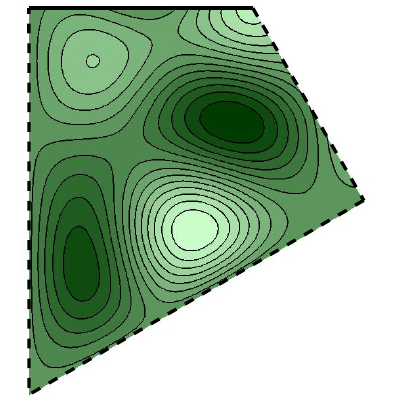}\hspace{6mm}\resizebox{3.0cm}{!}{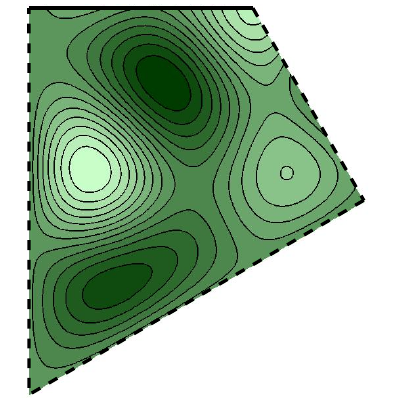}\hspace{6mm}\resizebox{3.0cm}{!}{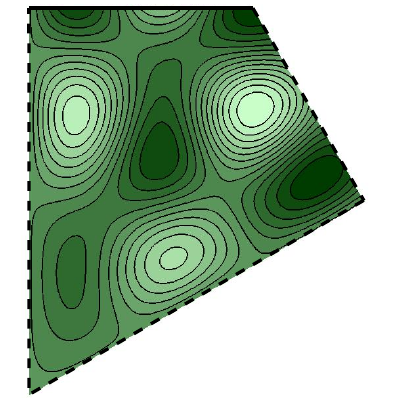}\hspace{6mm}\resizebox{3.0cm}{!}{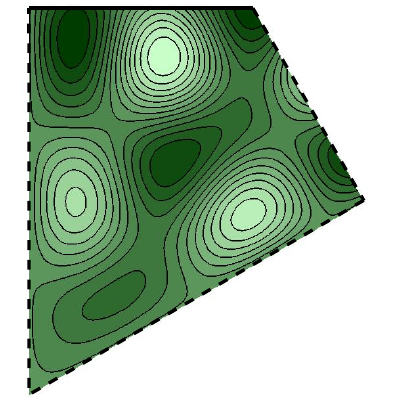}
\caption{\small The contour plots of Hartley orbit functions $\Hart^{\si^e,\si^l}$ over the fundamental domain $F^{\si^e,\si^l}$ of $G_2$. The dashed boundary does not belong to the fundamental domain.}\label{HElmG2}
\end{figure}

\end{example}
\section{Discretization of orbit functions}\setcounter{prop}{0}

\subsection{The point sets $F_M^{\wt{\si},\si}$ and $\Lambda_M^{\wt{\si},\si}$}\setcounter{prop}{0}\

Following the standard choice in Fourier analysis, only discrete values of labels of orbit functions $b\in P$ are considered. According to Propositions \ref{inv} and \ref{invH}, such discretization of labels induces the  $\wt\si-$antisymmetry of the arguments with respect to the affine Weyl subgroups $W_{\si}^{\aff}$ of both types of orbit functions $\Psi^{\wt{\si},\si}$ and $\Hart^{\wt{\si},\si}$. This allows to restrict the domain of orbit functions to the fundamental domain $F^{\one,\si}$  of  $W_{\si}^{\aff}$. Common zero points $H^{\wt{\si},\si}$ of the functions $\Psi^{\wt{\si},\si}$, $\Hart^{\wt{\si},\si}$ further restrict their domain to the sets  $F^{\wt\si,\si}$. Thus for any resolution factor $M\in \N$, the discrete Fourier calculus of orbit functions is developed on the set of points $F^{\wt{\si},\si}_M$,
\begin{equation*}%\label{FMc}
F^{\wt{\si},\si}_M=\frac{1}{M}P^{\vee}\cap F^{\wt{\si},\si}.
\end{equation*} 
The $W^\si-$invariance of the dual weight lattice $P^{\vee}$ and torus relations \eqref{rfun1E}, \eqref{rfun2E} imply the following crucial property of the discretized fundamental domains $F^{\one,\si}_M$,  
\begin{equation}\label{crucWF}
W^\si(F^{\one,\si}_M / Q^\vee )=\frac{1}{M}P^{\vee}/ Q^\vee.
\end{equation}   
 
The discretization of arguments of orbit functions induces in turn the label  $\wt{\si}-$antisymmetry from Propositions \ref{labelthm} and \ref{labelthmH}. This allows to restrict the labels of orbit functions to the magnified fundamental domain $MF^{\vee\one,\si}$ of the dual affine Weyl subgroups $\wh W_{\si}^{\aff}$. Zero orbit functions labelled by the the boundary points $MH^{\vee\wt{\si},\si}$ of $MF^{\vee\one,\si}$ are excluded yielding the sets of labels  
\begin{equation}\label{FMc}
\Lambda^{\wt{\si},\si}_M= P \cap MF^{\vee\wt{\si},\si}.
 \end{equation}
The discrete sets $F^{\wt{\si},\si}_M$ and  $\Lambda^{\wt{\si},\si}_M$ are explicitly constructed by utilizing structural results of  $F^{\wt{\si},\si}$ and $F^{\vee\wt{\si},\si}$ in Theorems \ref{thm123} and \ref{thm123d}. To this goal, simpler sets of points $F_M^{\si}$ and $\Lambda_M^{\si}$  from \cite{discliegrI,discliegrII} are recalled,
\begin{align*}
F_M^{\si} &=\frac{1}{M}P^{\vee}\cap F^{\si}, \\
\Lambda_M^{\si}&=P\cap MF\dual.
\end{align*}
Discretizing the relations \eqref{HF2} and \eqref{HFd2}, the sets $F^{\wt{\si},\si}_M$ and  $\Lambda^{\wt{\si},\si}_M$ are expressed as
\begin{align}
F_M^{\wt\si,\si}&= (F_M^{\wt\si}\cup  F_M^{\wt\si \cdot \si} ) \cup r_\si (F_M^{\wt\si}\cap  F_M^{\wt\si \cdot \si} ),\label{FMss} \\ 
\Lambda_M^{\wt\si,\si}&= (\Lambda_M^{\wt\si}\cup  \Lambda_M^{\wt\si \cdot \si} ) \cup r^{\vee}_\si (\Lambda_M^{\wt\si}\cap  \Lambda_M^{\wt\si \cdot \si} ).\label{LMss}
\end{align}
These structural relations are crucial for calculating the number of points in $F^{\wt{\si},\si}_M$ and  $\Lambda^{\wt{\si},\si}_M$.
\begin{tvr}\label{Fnumberofpoints}
Consider $\si\neq\one.$ For the numbers of points in $F_M^{\wt{\si},\si}$ and  $\Lambda_M^{\wt{\si},\si}$ it holds that
\begin{align}\label{numberofpoints}
\abs{F_M^{\wt{\si},\si}}&= \abs{F_M^{\wt{\si}}} +\abs{F_M^{\wt{\si}\cdot\si}},\\
\abs{\Lambda_M^{\wt{\si},\si}}&= \abs{\Lambda_M^{\wt{\si}}} +\abs{\Lambda_M^{\wt{\si}\cdot\si}}.\label{numberofpointsL}
\end{align}
\end{tvr}
\begin{proof}
Since the sets  $F_M^{\wt\si}\cup  F_M^{\wt\si \cdot \si}$  and $r_\si (F_M^{\wt\si}\cap  F_M^{\wt\si \cdot \si} )$ in \eqref{FMss} are disjoint it holds that 
\begin{equation}\label{disadd}
\abs{F_M^{\wt\si,\si}} = \abs {F_M^{\wt\si}\cup  F_M^{\wt\si \cdot \si}} +\abs{ r_\si (F_M^{\wt\si}\cap  F_M^{\wt\si \cdot \si} )}.
\end{equation}
For $\si\neq\one$ there exists some reflection $r_\si \in R^\si$ realizing a bijection and thus 
\begin{equation}\label{rR}
\abs{ r_\si (F_M^{\wt\si}\cap  F_M^{\wt\si \cdot \si} )} = \abs{F_M^{\wt\si}\cap  F_M^{\wt\si \cdot \si}}.
\end{equation}
Substituting the inclusion--exclusion principle relation
\begin{equation}
\abs {F_M^{\wt\si}\cup  F_M^{\wt\si \cdot \si}}=  \abs{F_M^{\wt{\si}}} +\abs{F_M^{\wt{\si}\cdot\si}} - \abs{F_M^{\wt\si}\cap  F_M^{\wt\si \cdot \si}}
\end{equation}
together with \eqref{rR} into \eqref{disadd} yields equation \eqref{numberofpoints}. Analogous steps lead from structure relation \eqref{LMss} to equation \eqref{numberofpointsL}.
\end{proof}
Formulas \eqref{numberofpoints} and \eqref{numberofpointsL} allow to derive
the following theorem which is crucial for completness relations and discrete transforms of orbit functions.
\begin{thm}\label{numelements} For any $\wt\si,\si\in \Sigma$ and $M\in \N$ it holds, for the numbers of elements of the sets  $F_M^{\wt{\si},\si}$ and $\Lambda_M^{\wt{\si},\si},$ that
\begin{equation}\label{comp}
\abs{ F_M^{\wt{\si},\si} }=\abs{\Lambda_M^{\wt{\si},\si} }.
\end{equation}
\end{thm}
\begin{proof}
Summarizing relations (30) and (36) in \cite{discliegrI} together with Corollary 5.4 in  \cite{discliegrII} yield, for any $\si \in \Sigma,$ the relation 
\begin{equation}\label{FMsLMs}
|F_M^\si|=|\Lambda_M^\si|.
\end{equation}
Since for $\si=\one$ the sets $\Lambda_M^{\wt{\si},\si}$ and $F_M^{\wt{\si},\si}$ in fact reduce to $F^{\wt\si}_M$ and $\Lambda^{\wt\si}_M$
$$F^{\wt\si,\one}_M= F^{\wt\si}_M,\q \Lambda^{\wt\si,\one}_M=\Lambda^{\wt\si}_M,$$
equation \eqref{comp} specializes in this case to \eqref{FMsLMs}. Substituting for  $\si\neq\one$ relation \eqref{FMsLMs} into
\eqref{numberofpoints} and \eqref{numberofpointsL} yields formula \eqref{comp}.
\end{proof}

%Recall that in \cite{discliegrI,discliegrII} are explicitly derived counting formulas for the number of elements for all cases of $F^{\one}_M$ and $\Lambda^{\one}_M$. Using the sign Coxeter number \eqref{msi}, Proposition 3.5 in \cite{discliegrI} and Proposition 5.1 in \cite{discliegrII} extend these results to all cases of $F_M^\si$ and $\Lambda_M^\si$,
%\begin{equation}
%|\Lambda_M^\si|=|F_M^\si| = \left\{ 
% \begin{array}{l l}
%  0 & \quad M < m^\si  \\
%   1 & \quad M=m^\si \\
%		|F^\one_{M-m^\si}| &\quad M>m^\si.
% \end{array} \right.	 \label{numsigma} 
%\end{equation}
%The following proposition calculates the number of points of $F_M^{\wt{\si},\si}$ using the numbers of
%$F_M^{\si}$ and $F_M^{\wt{\si}\cdot\si}.$ 

\subsection{Construction of the sets $F_M^{\wt{\si},\si}$ and $\Lambda_M^{\wt{\si},\si}$}\setcounter{prop}{0}\

An algorithm for the construction of the sets $F_M^{\wt{\si},\si}$ and $\Lambda_M^{\wt{\si},\si}$ relies on the structural relations \eqref{FMss} and~\eqref{LMss}. 
Note that according to \cite{discliegrI,discliegrII}, the sets $F_M^{\si}$ and $\Lambda_M^{\si}$ are described by defining the symbols $u_i^\si \in \R$, $i=0,\dots,n$ as 
\begin{alignat*}{2}
&u_i^\si \in \N, \quad &r_i &\in R^\si, \\
&u_i^\si \in \Z^{\geq 0}, \quad &r_i &\in R\setminus R^\si
\end{alignat*}
and their dual versions $t_i^{\si\vee} \in \R\dual,\, i=0,\cdots,n$ as
\begin{alignat*}{2}
&t_i^{\vee\si} \in \N, \quad &r_i \in R^{\vee\si}, \\
&t_i^{\vee\si} \in \Z^{\geq 0}, \quad &r_i \in R^{\vee} \setminus R^{\vee\si}.
\end{alignat*}
It holds that $R^\one = R^{\vee\one} =\emptyset$, $R^{\si^e} = R$ and lists of generators from $R^{\si^s}, R^{\si^l}$ and $R^{\vee{\si^s}}, R^{\vee{\si^l}}$ are contained in Table 1 in \cite{discliegrII}.  

The explicit form of $F_M^\si$ and $\Lambda_M^\si$ is then given by
\begin{alignat}{1}\label{definingrelation2}
F_M^\si &=\set {\frac{u_1^\si}{M}\wdual{1}+\cdots+\frac{u_n^\si}{M}\wdual{n}}{u_0^\si  + u_1^\si m_1 +\cdots + u_n^\si m_n = M},\\
\Lambda_M^{\si} &=\set {t_1^{\vee\si}\wvector{1}+\cdots+t_n^{\vee\si}{\wvector{n}}}
{t_0^{\vee\si}  + t_1^{\vee\si} m^\vee_1 +\cdots + t_n^{\vee\si} m^\vee_n = M}.\label{definingrelation2L}
\end{alignat}
Since $R^\one = \emptyset$, the four cases of the sets $F^{\si,\one}_M$ and $\Lambda^{\si,\one}_M$ in fact specialize to the sets $F^{\si}_M$ and $\Lambda^{\si}_M$,
 $$F^{\si,\one}_M= F^{\si}_M,\q \Lambda^{\si,\one}_M=\Lambda^{\si}_M.$$
For $\si \neq \one,$ there are three cases of $F^{\one,\si}_M$ and $\Lambda^{\one,\si}_M$ 
expressed as a union of two disjoint grid fragments: the point set $F^\one_M$ and the point set $r_{\si}  F_M^{\si}$ together with their dual counterparts,
\begin{equation}\label{Fonesi}
F^{\one,\si}_M=F^\one_M \cup r_{\si}  F_M^{\si}, \q \Lambda^{\one,\si}_M=\Lambda ^\one_M \cup r^\vee_{\si}  \Lambda_M^{\si}. 
\end{equation}
The next two choices of combinations of sign homomorphisms yield according to \eqref{FMss} and~\eqref{LMss} the relations
\begin{alignat}{1}
F_M^{\si^e, \si^s} &= F_M^{\si^l} \cup r_{\si^s} F_M^{\si^e}, \q \Lambda_M^{\si^e, \si^s} = \Lambda_M^{\si^l} \cup r^\vee_{\si^s} \Lambda_M^{\si^e} \label{FMes}\\
F_M^{\si^e, \si^l}& =F_M^{\si^s} \cup r_{\si^l} F_M^{\si^e}, \q \Lambda_M^{\si^e, \si^l} =\Lambda_M^{\si^s} \cup r^\vee_{\si^l} \Lambda_M^{\si^e}, \label{FMel}
\end{alignat}
and the last case is of the form
\begin{alignat}{1}
F_M^{\si^l, \si^e}& = (F_M^{\si^l}\cup F_M^{\si^s})\cup r_{\si^e}F_M^{\si^e}, \label{Fle} \\
\Lambda_M^{\si^l, \si^e}& = ( \Lambda_M^{\si^l}\cup \Lambda_M^{\si^s})\cup r^\vee_{\si^e}\Lambda_M^{\si^e}. \label{Lle}
\end{alignat}

%%%%%%%%%%%%%%%%%%%%%%%%%%%%%%%%%%%%%%%%%%%%%%%%%%%%%%%%%%%%%%%%%%
%%%%%%%%%%%%%%%%%%%%%%%%%%%%%%%%%%%%%%%%%%%%%%%%%%%%%%%%%%%%%%%%%%
%\subsection{Number of elements of $F_M^{\wt{\si},\si}$ and $\Lambda_M^{\wt{\si},\si}$}\setcounter{prop}{0}\
%%%%%%%%%%%%%%%%%%%%%%%%%%%%%%%%%%%%%%%%%%%%%%%%%%%%%%%%%%%%%%%%%%
%%%%%%%%%%%%%%%%%%%%%%%%%%%%%%%%%%%%%%%%%%%%%%%%%%%%%%%%%%%%%%%%%%

\normalsize
\begin{example}\label{ex1}
Consider the Lie algebra $C_2$ and the scaling factor $M=4$. %The order of the quotient group is $\abs{\frac{1}{4}P^{\vee}/Q^{\vee}}=32$.  
Firstly, using counting formulas from Appendix, the numbers of elements of $F_M^{\wt{\si},\si}$, $\si\neq\one$ are calculated as
\begin{alignat*}{2}
\abs{F^{\one,\si^s}_4}&=\comb{4}{2}+\comb{3}{2}+\comb{3}{2}+\comb{2}{2}= 13,\\  
\abs{F^{\one,\si^e}_4}&= \comb{4}{2}+\comb{3}{2}+\comb{2}{2}+\comb{1}{2}= 10, \\
\abs{F^{\one,\si^l}_4}&=\comb{4}{2}+\comb{3}{2}+\comb{3}{2}+\comb{2}{2} = 13, \\
\abs{F^{\si^e,\si^s}_4}&=\comb{3}{2}+\comb{2}{2}+\comb{2}{2}+\comb{1}{2}= 5, \\
\abs{F^{\si^l,\si^e}_4}&=\comb{3}{2}+\comb{2}{2}+\comb{3}{2}+\comb{2}{2}=8, \\
\abs{F^{\si^e,\si^l}_4}&= \comb{3}{2}+2\comb{2}{2}+\comb{1}{2}=5. 
\end{alignat*}

The nine points of the set $F^{\one}_4$, given by \eqref{definingrelation2},  are explicitly listed with coordinates in the orthonormal basis $(\al_1, \al_1 + \al_2)$ shown in Figure \ref{figC2},
\begin{alignat*}{2}
F^{\one}_4=& \left\{ \left[ \begin {array}{c} 0\\ 0\end {array} \right] 
, \left[ \begin {array}{c} 0\\  \frac{1}{4}\end {array}
 \right] , \left[ \begin {array}{c} 0\\  \frac{1}{2}
\end {array} \right] , \left[ \begin {array}{c} 0\\ \frac{3}{4}
\end {array} \right] , \left[ \begin {array}{c} 0
\\ 1\end {array} \right] , \left[ \begin {array}{c} \frac{1}{4}
\\  \frac{1}{4}\end {array} \right] , \left[ \begin {array}
{c}  \frac{1}{4}\\  \frac{1}{2}\end {array} \right] , \left[ 
\begin {array}{c}  \frac{1}{4}\\  \frac{3}{4}\end {array} \right] ,
 \left[ \begin {array}{c}  \frac{1}{2}\\  \frac{1}{2}\end {array}
 \right] \right\} .
\end{alignat*}
The  three sets $F^{\one,\si}_4$ are calculated using relation \eqref{Fonesi},
\begin{alignat*}{2}
F^{\one,\si^s}_4=& F^{\one}_4 \cup \left\{ \left[ \begin {array}{c} - \frac{1}{4}\\  \frac{1}{4}
\end {array} \right] , \left[ \begin {array}{c} - \frac{1}{4}
\\  \frac{1}{2}\end {array} \right] , \left[ \begin {array}
{c} - \frac{1}{4}\\  \frac{3}{4}\end {array} \right] , \left[ 
\begin {array}{c} - \frac{1}{2}\\  \frac{1}{2}\end {array} \right] \right\}, \\
F^{\one,\si^e}_4=& F^{\one}_4 \cup \left\{ \left[ \begin {array}{c} - \frac{1}{4}\\  \frac{1}{2}
\end {array} \right] \right\}, \\
F^{\one,\si^l}_4=& F^{\one}_4\cup \left\{ \left[ \begin {array}{c}  \frac{1}{4}\\ 0
\end {array} \right] , \left[ \begin {array}{c}  \frac{1}{2}
\\ 0\end {array} \right] , \left[ \begin {array}{c} \frac{3}{4}
\\ 0\end {array} \right] , \left[ \begin {array}
{c}  \frac{1}{2}\\  \frac{1}{4}\end {array} \right] \right\}.
\end{alignat*}
The last three cases of the sets $F^{\wt\si,\si}_4$ are calculated using relations \eqref{FMes}, \eqref{FMel} and \eqref{Fle},
\begin{alignat*}{2}
F^{\si^e,\si^s}_4&= \left\{ \left[ \begin {array}{c} 0\\  \frac{1}{4}\end {array}
 \right] , \left[ \begin {array}{c} 0\\  \frac{1}{2}
\end {array} \right] , \left[ \begin {array}{c} 0\\ \frac{3}{4}
\end {array} \right] , \left[ \begin {array}{c}  \frac{1}{4}
\\  \frac{1}{2}\end {array} \right] \right\} \cup \left\{ \left[ \begin {array}
{c} - \frac{1}{4}\\  \frac{1}{2}\end {array} \right] \right\},\\
F^{\si^l,\si^e}_4&= \left\{ \left[ \begin {array}{c} 0\\  \frac{1}{4}\end {array}
 \right] , \left[ \begin {array}{c} 0\\  \frac{1}{2}
\end {array} \right] , \left[ \begin {array}{c} 0\\ \frac{3}{4}
\end {array} \right] , \left[ \begin {array}{c}  \frac{1}{4}
\\  \frac{1}{4}\end {array} \right] , \left[ \begin {array}
{c}  \frac{1}{4}\\  \frac{1}{2}\end {array} \right] , \left[ 
\begin {array}{c}  \frac{1}{4}\\  \frac{3}{4}\end {array} \right] ,
 \left[ \begin {array}{c}  \frac{1}{2}\\  \frac{1}{2}\end {array}
 \right] \right\} \cup \left\{\left[ \begin {array}{c} - \frac{1}{4}\\  \frac{1}{2}
\end {array} \right] \right\},\\ 	
F^{\si^e,\si^l}_4&= \left\{ \left[ \begin {array}{c}  \frac{1}{4}\\  \frac{1}{4}\end {array}
 \right] , \left[ \begin {array}{c}  \frac{1}{4}\\  \frac{1}{2}
\end {array} \right] , \left[ \begin {array}{c}  \frac{1}{4}
\\  \frac{3}{4}\end {array} \right] , \left[ \begin {array}
{c}  \frac{1}{2}\\  \frac{1}{2}\end {array} \right]\right\} \cup \left\{ \left[ 
\begin {array}{c}  \frac{1}{2}\\  \frac{1}{4}\end {array} \right] \right\}.\\
\end{alignat*}
\normalsize
The coset representants of the quotient group $\frac{1}{4}P^{\vee}/Q^{\vee}$ together with the sets $F^{\wt{\si},\si}_4$ for all cases are depicted in Figure \ref{figC2}.
\begin{figure}%[h!]
\resizebox{12cm}{!}{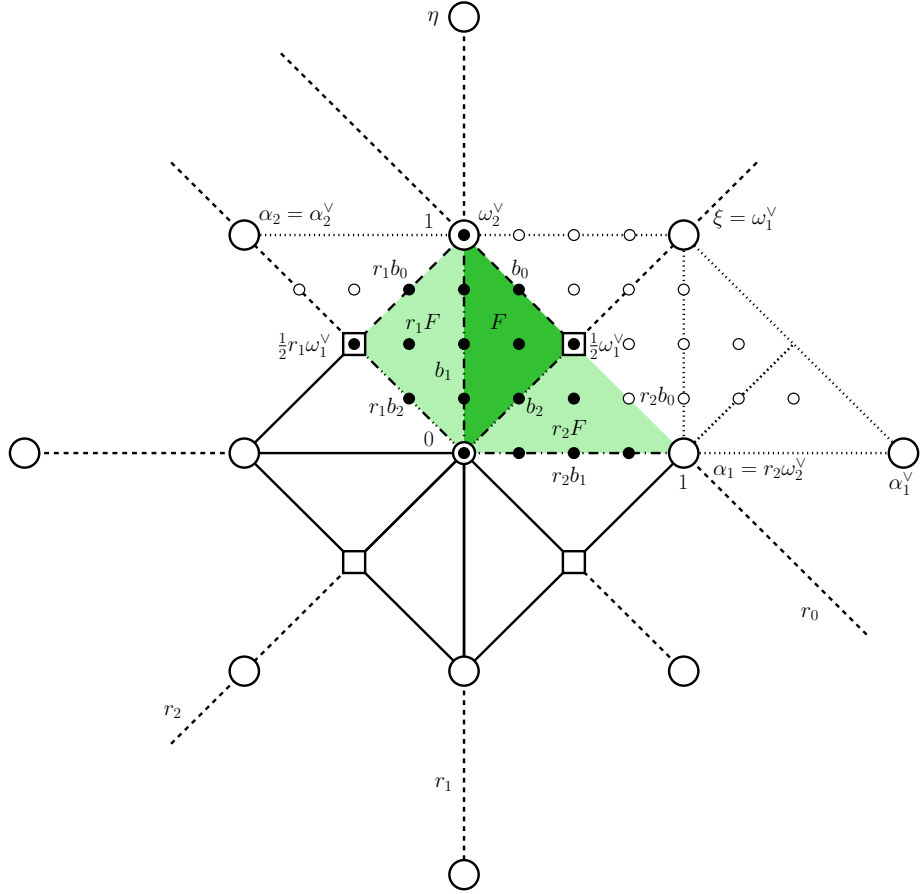}
\caption{\small Discrete sets $F^{\wt{\si},\si}_4$ of $C_2$. The $32$ coset representants of $\frac{1}{4}P^{\vee}/Q^{\vee}$ of $C_2$ are shown as $17$ black-filled and $15$ white-filled  dots. The darker green triangle is the fundamental domain $F$, two lighter triangles represent its reflections $r_1 F$ and $r_2 F$. The union $F\cup r_1 F$ contains $13$ elements of $F^{\one,\sigma^s}_4$; without all of its borders this union contains the $5$ points of $F^{\sigma^e,\sigma^s}_4$. Excluding the borders $r_1b_0$, $r_1 b_2$ and the point $\frac{1}{2}r_1 \om^\vee_1$ of $F\cup r_1 F$, the $10$ points of $F^{\one,\sigma^e}_4$ are obtained; excluding moreover the point $\om^\vee_2 $ and the point $0$, the $8$ points of $F^{\sigma^l,\sigma^e}_4$ are given. The union $F\cup r_2 F$ without the border $r_2b_0$ and the point $r_2\om_2^\vee$ contains $13$ elements of $F^{\one,\sigma^l}_4$. Leaving only the border $b_0$ and the point $\frac{1}{2}\om^\vee_1$ on the borders of $F\cup r_2 F$, the $5$ points in $F^{\sigma^e,\sigma^l}_4$ are obtained.
}\label{figC2}
\end{figure}
\end{example}
\begin{example}\label{exdual}
Consider again the Lie algebra $C_2$ and the scaling factor $M=4$. According to Theorem \ref{numelements} it holds that $\abs{F^{\wt{\si},\si}_4} = \abs{\Lambda^{\wt{\si},\si}_4}$, 
therefore the numbers of elements of of $\Lambda_M^{\wt{\si},\si}$, $\si\neq\one$ are as follows
\begin{alignat*}{2}
\abs{\Lambda^{\one,\si^s}_4}&= 13, \ \  \abs{\Lambda^{\si^e,\si^s}_4}&=5, \\
\abs{\Lambda^{\one,\si^e}_4}&= 10, \ \  \abs{\Lambda^{\si^l,\si^e}_4}&=8, \\
\abs{\Lambda^{\one,\si^l}_4}&= 13, \ \  \abs{\Lambda^{\si^e,\si^l}_4}&=5. \\
\end{alignat*}
The nine points of the set $\Lambda^{\one}_4$, given by \eqref{definingrelation2L},  are explicitly listed with coordinates in the orthonormal basis $(\al_1, \al_1 + \al_2)$ shown in Figure \ref{figC2d},
\begin{alignat*}{2}
\Lambda^{\one}_4=& \left\{ \left[ \begin {array}{c} 0\\ 0\end {array} \right] , \left[ \begin {array}{c} 0\\ 1\end {array} \right] 
, \left[ \begin {array}{c} 0\\ 2\end {array} \right] 
, \left[ \begin {array}{c}  \frac{1}{2}\\  \frac{1}{2}\end {array}
 \right] , \left[ \begin {array}{c}  \frac{1}{2}\\  \frac{3}{2}
\end {array} \right] , \left[ \begin {array}{c} 1\\ 1
\end {array} \right] , \left[ \begin {array}{c} 1\\ 2
\end {array} \right] , \left[ \begin {array}{c}  \frac{3}{2}
\\  \frac{3}{2}\end {array} \right] , \left[ \begin {array}
{c} 2\\ 2\end {array} \right] \right\} .
\end{alignat*}
The  three sets $\Lambda^{\one,\si}_4$ are calculated using relation \eqref{Fonesi},
\begin{alignat*}{2}
 {\Lambda^{\one,\si^s}_4}&= \Lambda^{\one}_4 \cup\left\{  \left[ \begin {array}
{c} - \frac{1}{2}\\  \frac{1}{2}\end {array} \right] , \left[ 
\begin {array}{c} - \frac{1}{2}\\  \frac{3}{2}\end {array} \right] ,
 \left[ \begin {array}{c} -1\\ 1\end {array} \right] 
, \left[ \begin {array}{c} - \frac{3}{2}\\  \frac{3}{2}\end {array}
 \right] \right\}    ,\\  
 {\Lambda^{\one,\si^e}_4}&= \Lambda^{\one}_4 \cup\left\{ \left[ \begin {array}
{c} - \frac{1}{2}\\  \frac{3}{2}\end {array} \right] \right\} , \\
 {\Lambda^{\one,\si^l}_4}&= \Lambda^{\one}_4 \cup\left\{  \left[ \begin {array}
{c} 1\\ 0\end {array} \right] , \left[ \begin {array}
{c} 2\\ 0\end {array} \right] , \left[ \begin {array}
{c}  \frac{3}{2}\\  \frac{1}{2}\end {array} \right] , \left[ 
\begin {array}{c} 2\\ 1\end {array} \right] \right\}
\end{alignat*}
The last three cases of the sets $\Lambda^{\wt\si,\si}_4$ are calculated using relations \eqref{FMes}, \eqref{FMel} and \eqref{Lle},
\begin{alignat*}{2}
 {\Lambda^{\si^e,\si^s}_4}&=\left\{ \left[ \begin {array}{c} 0\\ 1\end {array} \right] 
, \left[ \begin {array}{c} 0\\ 2\end {array} \right] 
, \left[ \begin {array}{c}  \frac{1}{2}\\  \frac{3}{2}\end {array}
 \right] , \left[ \begin {array}{c} 1\\ 2\end {array}
 \right] \right\} \cup \left\{ \left[ \begin {array}{c} - \frac{1}{2}\\  \frac{3}{2}
\end {array} \right] \right\}, \\
 {\Lambda^{\si^l,\si^e}_4}&= \left\{ \left[ \begin {array}{c} 0\\ 1\end {array} \right] 
, \left[ \begin {array}{c} 0\\ 2\end {array} \right] 
, \left[ \begin {array}{c}  \frac{1}{2}\\  \frac{1}{2}\end {array}
 \right] , \left[ \begin {array}{c}  \frac{1}{2}\\  \frac{3}{2}
\end {array} \right] , \left[ \begin {array}{c} 1\\ 1
\end {array} \right] , \left[ \begin {array}{c} 1\\ 2
\end {array} \right] , \left[ \begin {array}{c}  \frac{3}{2}
\\  \frac{3}{2}\end {array} \right] \right\} \cup \left\{ \left[ \begin {array}
{c} - \frac{1}{2}\\  \frac{3}{2}\end {array} \right] \right\} , \\
 {\Lambda^{\si^e,\si^l}_4}&=\left\{ \left[ \begin {array}{c}  \frac{1}{2}\\  \frac{1}{2}\end {array}
 \right] , \left[ \begin {array}{c}  \frac{1}{2}\\  \frac{3}{2}
\end {array} \right] , \left[ \begin {array}{c} 1\\ 1
\end {array} \right] , \left[ \begin {array}{c}  \frac{3}{2}
\\  \frac{3}{2}\end {array} \right] \right\} \cup \left\{ \left[ \begin {array}
{c}  \frac{3}{2}\\  \frac{1}{2}\end {array} \right] \right\}. \\
\end{alignat*}
\normalsize

The coset representants of the quotient group $P/4Q$ together with the sets $\Lambda^{\wt{\si},\si}_4$ for all cases are depicted in Figure~\ref{figC2d}.
\begin{figure}%[h!]
\resizebox{12cm}{!}{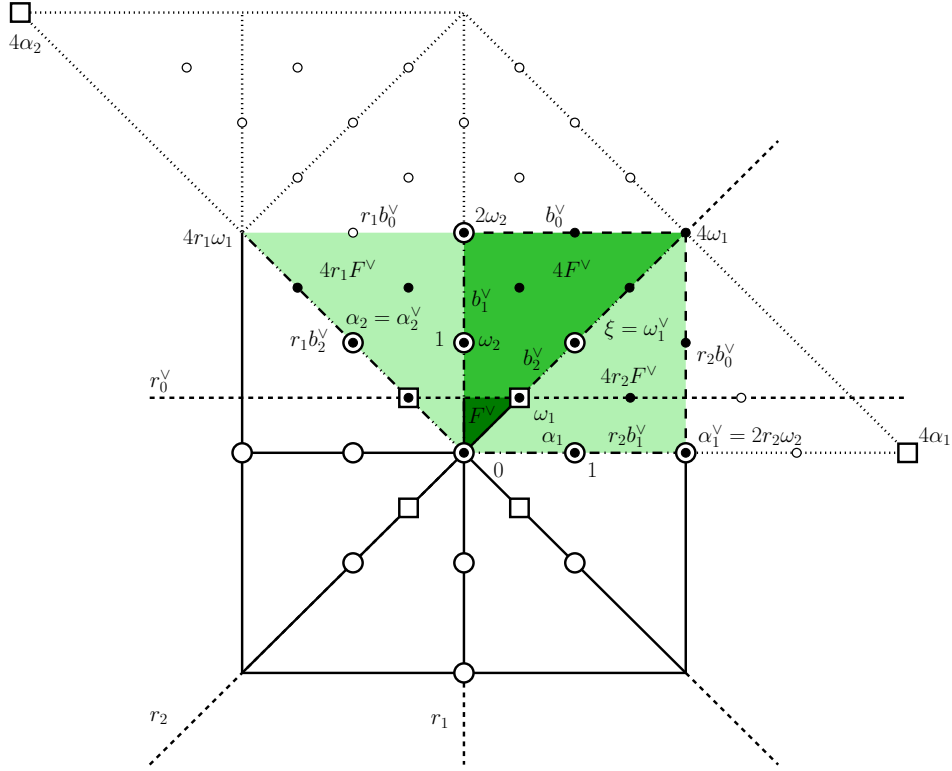}
\caption{\small Discrete sets $\Lambda^{\wt{\si},\si}_4$ of $C_2$. The $32$ coset representants of $P/4Q$ of $C_2$ are shown as $17$ black-filled and $15$ white-filled  dots. The darker green triangle is the magnified fundamental domain $4F^{\vee} $, two lighter triangles represent its reflections $4r_1 F^{\vee} $ and $4r_2 F^{\vee} $. The union $4F^{\vee} \cup 4r_1 F^{\vee} $ without the border $r_1b_0^{\vee} $ and the point $4r_1\om_1$ contains $13$ elements of $\Lambda^{ \one,\sigma^s}_4$. Leaving only the border $b_0^{\vee} $ and the point $2\om_2$ on the borders of $4F^{\vee} \cup 4r_1 F^{\vee} $ the $5$ points of $\Lambda^{ \sigma^e,\sigma^s}_4$ are obtained. Excluding the borders $r_1b_0^{\vee} $, $r_1 b_2^{\vee} $ and point $4r_1 \om_1$ of $4F^{\vee} \cup 4r_1 F^{\vee} $ the $10$ points of $\Lambda^{ \one,\sigma^e}_4$ are obtained; excluding moreover the point $4\om_1 $ and the point $0$, the $8$ points of $\Lambda^{ \sigma^e,\sigma^e}_4$ are given. The union $4F^{\vee} \cup 4r_2 F^{\vee} $ contains $13$ elements of $\Lambda^{ \one,\sigma^l}_4$; without all of its borders this union contains the $5$ points in $\Lambda^{ \sigma^e,\sigma^l}_4$. }\label{figC2d}
\end{figure}\end{example}

%%%%%%%%%%%%%%%%%%%%%%%%%%%%%%%%%%%%%%%%%%%%%%%%%%%%%%%%%%%%%%%%%%
%%%%%%%%%%%%%%%%%%%%%%%%%%%%%%%%%%%%%%%%%%%%%%%%%%%%%%%%%%%%%%%%%%
\subsection{Discrete orthogonality of orbit functions}\setcounter{prop}{0}\
%%%%%%%%%%%%%%%%%%%%%%%%%%%%%%%%%%%%%%%%%%%%%%%%%%%%%%%%%%%%%%%%%%
%%%%%%%%%%%%%%%%%%%%%%%%%%%%%%%%%%%%%%%%%%%%%%%%%%%%%%%%%%%%%%%%%%

The discrete orthogonality relations of the discretized functions $\Psi^{\wt{\si},\si}_b$, $b\in \Lambda_M^{\wt{\si},\si}$ on the finite point sets $F_M^{\wt{\si},\si}$ are formulated in this section.    
To this end, the discrete orthogonality relations of standard exponential functions on the discretized torus $\frac{1}{M}P^{\vee}/Q^{\vee}$ are needed \cite{discliegrI,discliegrII,MP2}. These hold for any $b \in P$ and state that
\begin{alignat*}{1}
\sum_{x \in \frac{1}{M}P^{\vee}/Q^{\vee}}{\e^{2\pi\im\scalar{b}{x}}} &= \left\{ 
  \begin{array}{l l}
    cM^n, & \quad \text{if $b \in MQ$}\\
    0\,, & \quad \text{otherwise.}
  \end{array} \right. 
\end{alignat*}
These relations can be equivalently expressed as
\begin{equation}\label{expDOG}
\sum_{x \in \frac{1}{M}P^{\vee}/Q^{\vee}}{\e^{2\pi\im\scalar{b}{x}}} = \sum_{q \in Q}{cM^n \delta_{b,Mq}}.
\end{equation}
The scalar products of any two functions $f,\,g : F_M^{\wt{\si},\si} \map \mathbbm{C}$ are introduced via the coefficients $\ep^{\si}(\x)$ determined by \eqref{epsi}, 
$$\scalar{f}{g}_{F_M^{\wt{\si},\si}}=\sum_{\x \in F_M^{\wt{\si},\si}} {\ep^{\si}(\x)f(\x)\overline{g(\x)}}.$$
The general discrete orthogonality relations of all ten types of orbit functions with respect to these scalar products are summarized in the following theorem. 
\begin{thm}\label{DOGorbitfunctions} 
For any $\si, \wt \si \in \Sigma$ and any $b,b'\in \Lambda_M^{\wt{\si},\si}$ it holds that
\begin{equation}\label{disk}
\scalar{\Psi_{b}^{\wt{\si},\si}}{\Psi_{b'}^{\wt{\si},\si}}_{F_M^{\wt{\si},\si}}= c\abs{W^{\si}}M^n h_M^{\vee\si}(b) \delta_{b,b'},
\end{equation}
where $c$, $h_M^{\vee\si}$ are defined by \eqref{Center} and \eqref{hM}, respectively, and $\abs{W^{\si}}$ is the order of the subgroup $W^{\si}$.  
\end{thm}
\begin{proof}
First, the common zero points of orbit functions from relation \eqref{zerop} together with the decompositions of the fundamental domains \eqref{Fsisi}, \eqref{disdec} yield
\begin{equation*}
\sum_{\x \in F_M^{\wt{\si},\si}} {\ep^{\si}(\x)\Psi_b^{\wt{\si},\si}(\x)\overline{\Psi_{b'}^{\wt{\si},\si}(\x)}} =\sum_{\x \in F_M^{\one,\si}} {\ep^{\si}(\x)\Psi_b^{\wt{\si},\si}(\x)\overline{\Psi_{b'}^{\wt{\si},\si}(\x)}}.
\end{equation*}
The expression ${\ep^{\si}(\x)\Psi_b^{\wt{\si},\si}(\x)\overline{\Psi_{b'}^{\wt{\si},\si}(\x)}}$ is in fact $W_{\si}^{\aff}-$invariant due to \eqref{epshift} and \eqref{xinvariant}. The implied invariance with respect to the shifts from $Q\dual$ together with \eqref{FD2}, \eqref{epsilontilda} give the relation
\begin{equation*}
\sum_{\x \in F_M^{\one,\si}} {\ep^{\si}(\x)\Psi_b^{\wt{\si},\si}(\x)\overline{\Psi_{b'}^{\wt{\si},\si}(\x)}} = \sum_{x \in F_M^{\one,\si}/Q\dual} {\wt{\ep}^{\si}(x)\Psi_b^{\wt{\si},\si}(x)\overline{\Psi_{b'}^{\wt{\si},\si}(x)}}.
\end{equation*}
The $W^{\si}-$invariance of  $\Psi_b^{\wt{\si},\si}(x)\overline{\Psi_{b'}^{\wt{\si},\si}(x)}$ and property \eqref{crucWF} allow changing the summation from $F_M^{\one,\si}/Q\dual$ to $ \frac{1}{M}P\dual/Q\dual$, i.e.     
\begin{equation*}
 \sum_{x \in F_M^{\one,\si}/Q\dual} {\wt{\ep}^{\si}(x)\Psi_b^{\wt{\si},\si}(x)\overline{\Psi_{b'}^{\wt{\si},\si}(x)}} =\sum_{x \in \frac{1}{M}P^{\vee}/Q^{\vee}} {\Psi_b^{\wt{\si},\si}(x)\overline{\Psi_{b'}^{\wt{\si},\si}(x)}}.
\end{equation*}
Next, formulas \eqref{genf} and \eqref{expDOG} are substituted,
\begin{align}
 \sum_{x \in \frac{1}{M}P^{\vee}/Q^{\vee}} {\Psi_b^{\wt{\si},\si}(x)\overline{\Psi_{b'}^{\wt{\si},\si}(x)}} &=\sum_{x \in \frac{1}{M}P^{\vee}/Q^{\vee}} \sum_{w' \in W^{\si}}\sum_{w \in W^{\si}}{\wt{\si}(w)\wt{\si}(w')\e^{2\pi\im\scalar{wb-w'b'}{x}}} \nonumber \\
 &=cM^n \sum_{w' \in W^{\si}} \sum_{w \in W^{\si}} \sum_{q \in Q}  {\wt{\si}(w)\wt{\si}(w') \delta_{wb,w'b'+Mq}} \nonumber \\
 &=cM^n \abs{W^{\si}}\sum_{w \in W^{\si}} \sum_{q \in Q}  {\wt{\si}(w) \delta_{b,wb'+Mq}}.  \label{deltaW}
 \end{align}
Consider the non-zero terms in summation \eqref{deltaW}.
The condition $b/M = w(b'/M)+q,$  due to \eqref{fun2Ed}, implies $w(b/M)+q=w^\aff (b/M)=b/M$ and $b/M=b'/M.$ Therefore, the summation set is reducible to $\Stab_{\wh{W}_{\si}^{\aff}}(b/M)$ and condition $b=b'$ holds, i.e.
$$cM^n \abs{W^{\si}}\sum_{w \in W^{\si}} \sum_{q \in Q}  {\wt{\si}(w) \delta_{b,wb'+Mq}} = 
cM^n\abs{W^{\si}}\delta_{b,b'}\sum_{w^{\aff} \in \Stab_{\wh{W}_{\si}^{\aff}}\left(\frac{b}{M}\right)} {\wt{\si}\circ \wh{\psi}(w^{\aff})}.$$
Finally, due to \eqref{FMc} it holds that $b/M \in F^{\vee\wt{\si},\si}$. Property \eqref{dualweightdomain} grants that $\wt{\si}\circ\wh{\psi}(w^{\aff})=1$ and the orthogonality relations are derived, 
 $$ cM^n\abs{W^{\si}}\delta_{b,b'}\sum_{w^{\aff} \in \Stab_{\wh{W}_{\si}^{\aff}}(b/M)} {\wt{\si}\circ \wh{\psi}(w^{\aff}) }=cM^n \abs{W^{\si}} h_M^{\vee\si}(b) \delta_{b,b'}.$$
\end{proof}
The discrete orthogonality relations of all types of functions $\Psi^{\wt{\si},\si}$ from Theorem \ref{DOGorbitfunctions} are also inherited by the related orbit functions with Hartley kernel $\Hart^{\wt{\si},\si}$. 
\begin{thm}\label{HartOG}
 For any $\si, \wt \si \in \Sigma$ and any $b,b'\in\Lambda_M^{\wt{\si},\si}$ it holds that
\begin{equation}\label{diskH}
\scalar{\Hart_{b}^{\wt{\si},\si}}{\Hart_{b'}^{\wt{\si},\si}}_{F_M^{\wt{\si},\si}}=c\abs{W^{\si}}M^n  h_M^{\vee\si}(b) \delta_{b,b'}.
\end{equation}
\end{thm}
\begin{proof}
An important consequence of \eqref{disk} is that for $b,b'\in \Lambda_M^{\wt{\si},\si}$ it holds that 
\begin{equation}\label{reality}
\Imim	\scalar{\Psi_{b}^{\wt{\si},\si}}{\Psi_{b'}^{\wt{\si},\si}}_{F_M^{\wt{\si},\si}} =0.
\end{equation}
The connection between the Hartley orbit functions and the regular orbit functions \eqref{HartS} and relation \eqref{reality} allow the following calculations,
\begin{align*}\scalar{\Hart_{b}^{\wt{\si},\si}}{\Hart_{b'}^{\wt{\si},\si}}_{F_M^{\wt{\si},\si}} =&
\scalar {\Rere \Psi_b^{\wt{\si},\si}+\Imim \Psi_b^{\wt{\si},\si}}{\Rere \Psi_{b'}^{\wt{\si},\si}+\Imim \Psi_{b'}^{\wt{\si},\si}}_{F_M^{\wt{\si},\si}}\\
 =& \frac{1}{4} \scalar {\left({\Psi_b^{\wt{\si},\si}+\overline{\Psi_b^{\wt{\si},\si}}}\right)-\im \left(\Psi_b^{\wt{\si},\si} - \overline{\Psi_b^{\wt{\si},\si}} \right) }{\left( {\Psi_{b'}^{\wt{\si},\si}+\overline{\Psi_{b'}^{\wt{\si},\si}}}\right)-\im\left({\Psi_{b'}^{\wt{\si},\si} - \overline{\Psi_{b'}^{\wt{\si},\si}}}\right)}_{F_M^{\wt{\si},\si}}. \\ =& \scalar{\Psi_{b}^{\wt{\si},\si}}{\Psi_{b'}^{\wt{\si},\si}}_{F_M^{\wt{\si},\si}} +\Imim \scalar{\Psi_{b}^{\wt{\si},\si}}{\overline{\Psi_{b'}^{\wt{\si},\si}}}_{F_M^{\wt{\si},\si}}.
\end{align*}
It remains to prove that
\begin{equation}\label{improve}
\Imim \scalar{\Psi_{b}^{\wt{\si},\si}}{\overline{\Psi_{b'}^{\wt{\si},\si}}}_{F_M^{\wt{\si},\si}} = 0.
\end{equation} 
First, from definition \eqref{genf} follows directly the relation for the complex conjugated function, 
\begin{equation}\label{complex}
\overline{\Psi_{b'}^{\wt{\si},\si}} = {\Psi_{-b'}^{\wt{\si},\si}}.
\end{equation} 
Property \eqref{fun1Ed} of the fundamental domain $F^{\vee\one,\si}$ of  $\wh{W}_{\si}^{\aff}$ implies that there exist a point $\wt{b} \in F^{\vee\one,\si}$ and $\wh{w}_{\si}^{\aff} \in \wh{W}_{\si}^{\aff} $ such that $-b'/M = \wh{w}_{\si}^{\aff}\wt{b}.$
Since $-b'/M  \in \frac{1}{M}P$ and $\frac{1}{M}P$ is $\wh{W}_{\si}^{\aff}-$invariant, one has that $\wt{b}\in\frac{1}{M}P $ and thus $b''= M\wt{b}\in P \cap MF^{\vee\one,\si}$. 
From the resulting relation
$$-b'= M\wh{w}_{\si}^{\aff}\left(\frac{b''}{M}\right), $$ 
equation \eqref{complex} and Proposition \ref{labelthm}, it follows that
\begin{equation}\label{complex2}\Imim \scalar{\Psi_{b}^{\wt{\si},\si}}{\overline{\Psi_{b'}^{\wt{\si},\si}}}_{F_M^{\wt{\si},\si}}=\Imim \scalar{\Psi_{b}^{\wt{\si},\si}}{{\Psi_{M\wh{w}_{\si}^{\aff}\left(\frac{b''}{M}\right)}^{\wt{\si},\si}}}_{F_M^{\wt{\si},\si}}= \wt{\si}\circ\wh{\psi}(\wh{w}_{\si}^{\aff})\cdot  \Imim \scalar{\Psi_{b}^{\wt{\si},\si}}{{\Psi_{b''}^{\wt{\si},\si}}}_{F_M^{\wt{\si},\si}}.\end{equation}
Secondly, if $b'' \in   MF^{\vee \one,\si}\setminus MF^{\vee \wt{\si},\si} = MH^{\vee\wt{\si},\si},$ 
then Proposition \ref{labelthm} states that the orbit function  $\Psi_{b''}^{\wt{\si},\si}$ vanishes on $F_M^{\wt{\si},\si}$ and  \eqref{improve} is thus valid.
If $b'' \in   MF^{\vee \wt{\si},\si}$, then indeed $b'' \in \Lambda_M^{\wt{\si},\si}$ and relation \eqref{reality} together with \eqref{complex2} imply that the scalar product in \eqref{improve} has no imaginary part.
\end{proof}

%%%%%%%%%%%%%%%%%%%%%%%%%%%%%%%%%%%%%%%
\subsection{Discrete orbit function transforms}\setcounter{prop}{0}\
%%%%%%%%%%%%%%%%%%%%%%%%%%%%%%%%%%%%%%%%%%%%%

%The set $\set{\Psi^{\wt{\si},\si}_b }{b \in \Lambda_M^{\wt{\si},\si}}$ 
%forms an orthogonal basis of a $\abs{F_M^{\wt{\si} ,\si}}$ dimensional space, due to 
%$$\abs{F_M^{\wt{\si} ,\si}} = \abs{ \Lambda_M^{\wt{\si},\si}}.$$

An arbitrary function $f:\R^n \map \mathbb{C}$, sampled on the point set $F_M^{\wt{\si},\si}$, can be interpolated by the interpolating function $\Int [f]^{\wt{\si},\si}_M$. The interpolating function $\Int[f]^{\wt{\si},\si}_M$ is required to coincide with $f$ at all the gridpoints of $F_M^{\wt{\si},\si},$
\begin{equation}\label{trans0}
\Int[f]^{\wt{\si},\si}_M(\x)=f(\x) ,\q\x \in F_M^{\wt{\si},\si}.\end{equation}
The interpolating function $\Int [f]^{\wt{\si},\si}_M$ is given in terms of expansion functions $\Psi_b^{\wt{\si},\si},$

\begin{equation}\label{trans1}
\Int[f]_M^{\wt{\si},\si}(\x)= \sum_{b \in \Lambda_M^{\wt{\si},\si}}{k^{\wt{\si},\si}_b \Psi_b^{\wt{\si},\si}(\x)}, \q\x \in \R^n.	
\end{equation}
The frequency spectrum coefficients $k^{\wt{\si},\si}_b$ are, due to Theorem \ref{DOGorbitfunctions} and Theorem \ref{numelements}, uniquely determined by the standard method of calculation of Fourier coefficients
\begin{equation}\label{trans2}k^{\wt{\si},\si}_b =\frac{\scalar{f}{\Psi_{b}^{\wt{\si},\si}}_{F_M^{\wt{\si},\si}}}{\scalar{{\Psi_{b}^{\wt{\si},\si}}}{\Psi_{b}^{\wt{\si},\si}}_{F_M^{\wt{\si},\si}}} = \frac{1}{c\abs{W^{\si}}M^n{h^{\vee\si}_M(b)}}\sum_{\x \in F^{\wt{\si},\si}_M}{\epsilon^{\si}(\x)f(\x) \overline{\Psi_b^{\wt{\si},\si}(\x)    }}.\end{equation}
Taking into account equality \eqref{trans0}, relations \eqref{trans2} and \eqref{trans1} constitute the forward and backward discrete Fourier-Weyl transforms, respectively, of the discretized function $f$. Furthermore, using the Parseval equality of the orthogonal basis $\Psi^{\wt{\si},\si}_b,  b \in \Lambda_M^{\wt{\si},\si}$ results in the following relation
$$\sum_{\x \in F_M^{\wt{\si},\si}} {\epsilon^{\si}(\x) \abs{f(\x)}^{2} } = c \abs{W^{\si}} M^n\sum_{b \in \Lambda_M^{\wt{\si},\si}} {h_M^{\vee \si}(b) \abs{k^{\wt{\si},\si}_b}^2}.$$

%%%%%%%%%%%%%%%%%%%%%%%%%%%%%%%%%%%%%%%%%%%%%%%%%%%%%%%%%%%%%%%%%%%%
%The set $\set{\Hart^{\wt{\si},\si}_b }{b \in \Lambda_M^{\wt{\si},\si}}$
%forms an orthogonal basis of a $\abs{F_M^{\wt{\si} ,\si}}$ dimensional space, due to 
%$$\abs{F_M^{\wt{\si} ,\si}} = \abs{ \Lambda_M^{\wt{\si},\si}}.$$

Similarly to the interpolation formulas and discrete transforms of the standard orbit functions, their related real-valued versions are formulated in terms of Hartley orbit functions. An arbitrary real-valued function $g:\R^n \map \R$ sampled on the point set $F_M^{\wt{\si},\si}$ can be interpolated by the real-valued interpolating functions $\Inth[g]^{\wt{\si},\si}_M$. Again, the interpolating function $\Inth[g]^{\wt{\si},\si}_M$ coincides with $g$ at all the gridpoints $F_M^{\wt{\si},\si},$ 
$$\Inth[g]^{\wt{\si},\si}_M(\x)=g(\x), \q  \x \in F_M^{\wt{\si},\si},$$
and is given in terms of expansion functions $\Hart_b^{\wt{\si},\si},$
$$\Inth[g]^{\wt{\si},\si}_M(a)= \sum_{b \in \Lambda_M^{\wt{\si},\si}}{\hartcoefficient^{\wt{\si},\si}_b \Hart_b^{\wt{\si},\si}(\x)}, \q a\in 	\R^n .$$
The frequency spectrum coefficients $\hartcoefficient^{\wt{\si},\si}_b$ of the Hartley-Weyl transform are determined by
\begin{equation}\label{trans2H}
\hartcoefficient^{\wt{\si},\si}_b =\frac{\scalar{g}{\Hart_{b}^{\wt{\si},\si}}_{F_M^{\wt{\si},\si}}}{\scalar{{\Hart_{b}^{\wt{\si},\si}}}{\Hart_{b}^{\wt{\si},\si}}_{F_M^{\wt{\si},\si}}} = \frac{1}{c\abs{W^{\si}}M^n{h^{\vee\si}_M(b)}}\sum_{\x \in F^{\wt{\si},\si}_M}{ \epsilon^{\si}(\x) g(\x) {\Hart_b^{\wt{\si},\si}(\x)    }}
\end{equation}
and the relation between the sum of squared values of $g$ and the sum of squared values of its frequency spectrum is  
 $$\sum_{\x \in F_M^{\wt{\si},\si}} {\epsilon^{\si}(\x) g^2(\x) } = c \abs{W^{\si}}M^n \sum_{b \in \Lambda_M^{\wt{\si},\si}} {h_M^{\vee \si}(b) (\hartcoefficient^{\wt{\si},\si}_b)^2}.$$

%%%%%%%%%%%%%%%%%%%%%%%%%%%%%%%%%%%%%%%
\section{Concluding Remarks}
\begin{itemize}
\item General form of product-to-sum decomposition formulas exists for any two functions $\Psi_b^{\si_1,{\si}}$, $\Psi_b^{\si_2,{\si}}$which correspond to an identical underlying subgroup $W^{\si}$, 
\begin{equation}\label{decomp}
\Psi_b^{\si_1,{\si}}(\x)\cdot \Psi_{b}^{\si_2,{\si}}(\x') =\sum_{w \in W^{\si}} \si_2 (w)\cdot\Psi_{b}^{\si_1\cdot \si_2,{\si}}(\x+w\x').
\end{equation}
This formula with $\abs{W^{\si}}$ additive terms covers 5 cases of possible decompositions of products for algebras with one length of roots and 28 cases for algebras with two lengths of roots. Similar relations with $4\abs{W^{\si}}$ additive terms also hold for the Hartley orbit functions.   
\item Discrete orthogonality relations \eqref{disk} and decomposition formulas \eqref{decomp} are in \cite{Ef2d} exemplified for six types of two-variable $E-$functions of algebras $C_2$ and $G_2$. Effectiveness of interpolation formulas \eqref{trans1} of these two-variable $E-$functions is demonstrated on complex-valued model functions in \cite{Ef2di}.  Comparable interpolating ability of real-valued functions is expected for Hartley orbit functions. Good performance of orbit functions in interpolation tasks indicates great potential in other fields related to digital data processing. The interpolation properties of all types of orbit functions as well as existence of general convergence criteria of the operator sequence $\Int_M^{\wt{\si},\si}: f\mapsto \Int[f]_M^{\wt{\si},\si} $, similar to the estimates for the $A_n$ case achieved in \cite{xuAd}, deserve further study. 
\item Link between the Weyl-orbit functions and the inherited discrete and continuous orthogonality relations of the generalized multidimensional Chebyshev polynomials is recently investigated in connection with the corresponding polynomial methods such as polynomial interpolation, approximation and cubature formulas \cite{xuAd,SsSlcub}. Cubature rules related to (anti)symmetric, short and long orbit functions are developed in full generality in \cite{HMcub,SsSlcub,MP4}. Extension of these methods to the six types of $E-$functions requires a concept of generalized Laurent-type polynomials  and poses an open problem.
\item The discrete transforms \eqref{trans2} and \eqref{trans2H} of orbit functions specialize for the case $A_1$ to one-variable discrete Fourier, discrete Hartley, discrete cosine and sine transforms \cite{Brac, Poul}. Except for the Hartley transform, straightforward concatenating of one-dimensional transforms is commonly used to create their 2D and 3D versions \cite{Poul}. This approach results in multitude of possible combinations of transforms in higher dimensions and in the context of Weyl orbit functions of semisimple Lie algebras is elaborated in \cite{HKP}. Moreover,  two fundamentally different choices of even Weyl groups $W^{\si^e}$ are also brought forward \cite{HKP}. The first option is to create an even Weyl group of a semisimple Lie algebra as the direct product of Weyl groups of its simple components; the second option is to introduce a full even Weyl group of the semisimple Lie algebras and their related reducible root systems. This method opens numerous further prospects to produce different types of short and long even Weyl subgroups of semisimple Lie algebras and their linked discrete transforms. Even more additional alternatives are obtainable by including, in this method, the shifted lattice transforms from \cite{CzHr}.
\item Discrete orthogonality relations \eqref{disk} and \eqref{diskH} are formulated on the points of the refined dual weight lattice. This choice of the points induces in turn the dual affine Weyl group (anti)symmetry of the orbit function labels in Propositions \ref{labelthm} and \ref{labelthmH}. Discrete transforms of the (anti)symmetric, short and long orbit functions on the points from the refined weight lattice are in connection with the conformal field theory solved in \cite{HW2}. The labels of this discretization share the same (anti)symmetry with the points generated by the given affine Weyl group.  The Fourier transforms constructed on the points of the refined (dual) root lattice represent the remaining unresolved discrete transforms related to the four classical Weyl group invariant lattices. The merit of having all four classical lattice transforms available is the possibility of generating novel and relevant transforms on generalized lattices, including the 2D honeycomb lattice. The open problem of detailing the root lattice transforms is however, specifically challenging, since the symmetry groups of the labels are not in general Coxeter groups.

\end{itemize}

%%%%%%%%%%%%%%%%%%%%%%%%%%%%%%%%%%%%%%%%%%%%%
\section*{Acknowledgments}
JH and MJ are very grateful for the hospitality extended to them at the Centre de recherches math\'ematiques, Universit\'e de Montr\'eal and for valuable consultations with J. Patera. This work was supported by the Grant Agency of the Czech Technical University in Prague, grant number SGS16/239/OHK4/3T/14. JH gratefully acknowledges the support of this work by RVO14000.

\section*{Appendix}

%The explicit counting formulas for the numbers of points of $F_M^{\wt{\si},\si}$ complete the counting formulas from \cite{discliegrI,discliegrII,discliegrE}.,
Completing the counting formulas from \cite{discliegrI,discliegrII,discliegrE},
the number of elements of $F_M^{\wt{\si},\si}$ for Lie algebras with two lengths of roots are given by the following relations.
\begin{enumerate}\label{numberofelementsC2}\item  $C_n,\,n\geq 2$,
\begin{align*}
|F_{2k}^{\one,\si^s}(C_n)|&=  \begin{pmatrix}k+n \\ n \end{pmatrix}+\begin{pmatrix}k+n-1 \\ n \end{pmatrix}+\begin{pmatrix}k+1 \\ n \end{pmatrix}+\begin{pmatrix}k \\ n \end{pmatrix}\\
|F_{2k+1}^{\one,\si^s}(C_n)|&= 2\begin{pmatrix}k+n \\ n \end{pmatrix}+2\begin{pmatrix}k+1 \\ n \end{pmatrix} ,\\
|F_{2k}^{\one,\si^l}(C_n)|&= \begin{pmatrix}k+n \\ n \end{pmatrix}+2\begin{pmatrix}n+k-1 \\ n \end{pmatrix}+\begin{pmatrix}n+k-2 \\ n \end{pmatrix}\\
|F_{2k+1}^{\one,\si^l}(C_n)|&= 2\begin{pmatrix}k+n \\ n \end{pmatrix}+2\begin{pmatrix}n+k-1 \\ n \end{pmatrix},\\
|F_{2k}^{\si^e,\si^s}(C_n)|&= \begin{pmatrix}n+k-1 \\ n \end{pmatrix}+\begin{pmatrix}n+k-2 \\ n \end{pmatrix} + \begin{pmatrix}k \\ n \end{pmatrix}+\begin{pmatrix}k-1 \\ n \end{pmatrix}\\
|F_{2k+1}^{\si^e,\si^s}(C_n)|&= 2\begin{pmatrix}n+k-1 \\ n \end{pmatrix} + 2\begin{pmatrix}k \\ n \end{pmatrix},\\
|F_{2k}^{\si^e,\si^l}(C_n)|&= \begin{pmatrix}k+1 \\ n \end{pmatrix}+ 2\begin{pmatrix}k \\ n \end{pmatrix}+\begin{pmatrix}k-1 \\ n \end{pmatrix}\\
|F_{2k+1}^{\si^e,\si^l}(C_n)|&= 2\begin{pmatrix}k+1 \\ n \end{pmatrix} + 2\begin{pmatrix}k \\ n \end{pmatrix},\\
|F_{2k}^{\si^l,\si^e}(C_n)|&= \begin{pmatrix}k+1 \\ n \end{pmatrix}+\begin{pmatrix}k \\ n \end{pmatrix}+\begin{pmatrix}n+k-1 \\ n \end{pmatrix}+\begin{pmatrix}n+k-2 \\ n \end{pmatrix}\\
|F_{2k+1}^{\si^l,\si^e}(C_n)|&= 2\begin{pmatrix}n+k-1 \\ n \end{pmatrix} + 2\begin{pmatrix}k+1 \\ n \end{pmatrix},
\end{align*}
\item $B_n,\,n\geq 3$,

\begin{align*}
|F_M^{\one,\si^s}(B_n)|&=|F_M^{\one,\si^l}(C_n)|,\\
|F_M^{\one,\si^l}(B_n)|&=|F_M^{\one,\si^s}(C_n)|,\\
|F_M^{\si^e,\si^s}(B_n)|&=|F_M^{\si^e,\si^l}(C_n)|,\\
|F_M^{\si^e,\si^l}(B_n)|&=|F_M^{\si^e,\si^s}(C_n)|,\\
|F_M^{\si^l,\si^e}(B_n)|&=|F_M^{\si^l,\si^e}(C_n)|,
\end{align*} 
 \small{
\item $G_2,$
\begin{equation*}
\begin{alignedat}{4}
|F_{6k}^{\one,\si^s}(G_2)|&= 1+3\,k+6\,{k}^{2}, &\qquad |F_{6k+1}^{\one,\si^s}(G_2)|&=1+5\,k+6\,{k}^{2}\\
|F_{6k+2}^{\one,\si^s}(G_2)|&= 2+7\,k+6\,{k}^{2}, &\qquad |F_{6k+3}^{\one,\si^s}(G_2)|&=4+9\,k+6\,{k}^{2}\\
|F_{6k+4}^{\one,\si^s}(G_2)|&=5+11\,k+6\,{k}^{2},  &\qquad |F_{6k+5}^{\one,\si^s}(G_2)|&=7+13\,k+6\,{k}^{2}\\
|F_{6k}^{\one,\si^e}(G_2)|&= 2+6\,{k}^{2}, &\qquad |F_{6k+1}^{\one,\si^e}(G_2)|&=1+2\,k+6\,{k}^{2}\\
|F_{6k+2}^{\one,\si^e}(G_2)|&= 2+4\,k+6\,{k}^{2}, &\qquad |F_{6k+3}^{\one,\si^e}(G_2)|&=3+6\,k+6\,{k}^{2}\\
|F_{6k+4}^{\one,\si^e}(G_2)|&=4+8\,k+6\,{k}^{2},  &\qquad |F_{6k+5}^{\one,\si^e}(G_2)|&=5+10\,k+6\,{k}^{2}\\
\end{alignedat}
\end{equation*}
\begin{equation*}
\begin{alignedat}{4}
|F_{6k}^{\si^l,\si^s}(G_2)|&= 1-3\,k+6\,{k}^{2}, &\qquad |F_{6k+1}^{\si^l,\si^s}(G_2)|&=-k+6\,{k}^{2}\\
|F_{6k+2}^{\si^l,\si^s}(G_2)|&= k+6\,{k}^{2}, &\qquad |F_{6k+3}^{\si^l,\si^s}(G_2)|&=1+3\,k+6\,{k}^{2}\\
|F_{6k+4}^{\si^l,\si^s}(G_2)|&=1+5\,k+6\,{k}^{2},  &\qquad |F_{6k+5}^{\si^l,\si^s}(G_2)|&=2+7\,k+6\,{k}^{2}\\
|F_{6k}^{\si^l,\si^e}(G_2)|&= 6\,k^2, &\qquad |F_{6k+1}^{\si^l,\si^e}(G_2)|&=2\,k+6\,{k}^{2}\\
|F_{6k+2}^{\si^l,\si^e}(G_2)|&= 4\,k+6\,{k}^{2}, &\qquad |F_{6k+3}^{\si^l,\si^e}(G_2)|&=2+6\,k+6\,{k}^{2}\\
|F_{6k+4}^{\si^l,\si^e}(G_2)|&=2+8\,k+6\,{k}^{2},  &\qquad |F_{6k+5}^{\si^l,\si^e}(G_2)|&=4+10\,k+6\,{k}^{2},
\end{alignedat}
\end{equation*}
\begin{equation*}
\begin{alignedat}{1}
|F_M^{\one,\si^l}(G_2)|&=|F_M^{\one,\si^s}(G_2)|,\\
|F_M^{\si^s,\si^l}(G_2)|&=|F_M^{\si^l,\si^s}(G_2)|,
\end{alignedat}
\end{equation*}
\item $F_4,$
\begin{equation*}
\begin{alignedat}{4}
|F_{12k}^{\one,\si^s}(F_4)|&=1+8\,k+25\,{k}^{2}+36\,{k}^{3}+36\,{k}^{4}, &\qquad |F_{12k+1}^{\one,\si^s} (F_4)|&= 1+10\,k+31\,{k}^{2}+48\,{k}^{3}+36\,{k}^{4}\\
|F_{12k+2}^{\one,\si^s}(F_4)|&= 3+20\,k+49\,{k}^{2}+60\,{k}^{3}+36\,{k}^{4}, &\qquad |F_{12k+3}^{\one,\si^s}(F_4)|&= 4+25\,k+61\,{k}^{2}+72\,{k}^{3}+36\,{k}^{4}\\
|F_{12k+4}^{\one,\si^s}(F_4)|&=8+42\,k+85\,{k}^{2}+84\,{k}^{3}+36\,{k}^{4},  &\qquad |F_{12k+5}^{\one,\si^s}(F_4)|&= 10+52\,k+103\,{k}^{2}+96\,{k}^{3}+36\,{k}^{4}\\
|F_{12k+6}^{\one,\si^s}(F_4)|&=18+78\,k+133\,{k}^{2}+108\,{k}^{3}+36\,{k}^{4},  &\qquad |F_{12k+7}^{\one,\si^s}(F_4)|&= 22+95\,k+157\,{k}^{2}+78\,{k}^{3}+36\,{k}^{4}\\
|F_{12k+8}^{\one,\si^s}(F_4)|&= 35+132\,k+193\,{k}^{2}+132\,{k}^{3}+36\,{k}^{4},  &\qquad |F_{12k+9}^{\one,\si^s}(F_4)|&=  43+158\,k+223\,{k}^{2}+144\,{k}^{3}+36\,{k}^{4}\\
|F_{12k+10}^{\one,\si^s}(F_4)|&= 63+208\,k+265\,{k}^{2}+156\,{k}^{3}+36\,{k}^{4},  &\qquad |F_{12k+11}^{\one,\si^s}(F_4)|&=76+245\,k+301\,{k}^{2}+168\,{k}^{3}+36\,{k}^{4}\\
|F_{12k}^{\one,\si^e}(F_4)|&=2+52\,{k}^{2}+36\,{k}^{4}, &\qquad |F_{12k+1}^{\one,\si^e} (F_4)|&= 1+8\,k+49\,{k}^{2}+12\,{k}^{3}+36\,{k}^{4}\\
|F_{12k+2}^{\one,\si^e}(F_4)|&= 3+18\,k+58\,{k}^{2}+24\,{k}^{3}+36\,{k}^{4}, &\qquad |F_{12k+3}^{\one,\si^e}(F_4)|&= 4+26\,k+61\,{k}^{2}+36\,{k}^{3}+36\,{k}^{4}\\
|F_{12k+4}^{\one,\si^e}(F_4)|&=8+40\,k+76\,{k}^{2}+48\,{k}^{3}+36\,{k}^{4},  &\qquad |F_{12k+5}^{\one,\si^e}(F_4)|&= 10+50\,k+85\,{k}^{2}+60\,{k}^{3}+36\,{k}^{4}\\
|F_{12k+6}^{\one,\si^e}(F_4)|&=17+70\,k+106\,{k}^{2}+72\,{k}^{3}+36\,{k}^{4}  &\qquad |F_{12k+7}^{\one,\si^e}(F_4)|&= 21+84\,k+121\,{k}^{2}+84\,{k}^{3}+36\,{k}^{4}\\
|F_{12k+8}^{\one,\si^e}(F_4)|&= 32+112\,k+148\,{k}^{2}+96\,{k}^{3}+36\,{k}^{4},  &\qquad |F_{12k+9}^{\one,\si^e}(F_4)|&=  39+132\,k+169\,{k}^{2}+108\,{k}^{3}+36\,{k}^{4}\\
|F_{12k+10}^{\one,\si^e}(F_4)|&= 55+170\,k+202\,{k}^{2}+120\,{k}^{3}+36\,{k}^{4},  &\qquad |F_{12k+11}^{\one,\si^e}(F_4)|&=66+198\,k+229\,{k}^{2}+132\,{k}^{3}+36\,{k}^{4}\\
|F_{12k}^{\si^l,\si^s}(F_4)|&=1-8\,k+25\,{k}^{2}-36\,{k}^{3}+36\,{k}^{4}, &\qquad |F_{12k+1}^{\si^l,\si^s} (F_4)|&= -3\,k+13\,{k}^{2}-24\,{k}^{3}+36\,{k}^{4}\\
|F_{12k+2}^{\si^l,\si^s}(F_4)|&= -2\,k+13\,{k}^{2}-12\,{k}^{3}+36\,{k}^{4}, &\qquad |F_{12k+3}^{\si^l,\si^s}(F_4)|&= 7\,{k}^{2}+36\,{k}^{4}\\
|F_{12k+4}^{\si^l,\si^s}(F_4)|&=2\,k+13\,{k}^{2}+12\,{k}^{3}+36\,{k}^{4},  &\qquad |F_{12k+5}^{\si^l,\si^s}(F_4)|&= 3\,k+13\,{k}^{2}+24\,{k}^{3}+36\,{k}^{4}\\
|F_{12k+6}^{\si^l,\si^s}(F_4)|&=1+8\,k+25\,{k}^{2}+36\,{k}^{3}+36\,{k}^{4},  &\qquad |F_{12k+7}^{\si^l,\si^s}(F_4)|&= 1+10\,k+31\,{k}^{2}+48\,{k}^{3}+36\,{k}^{4}\\
|F_{12k+8}^{\si^l,\si^s}(F_4)|&= 3+20\,k+49\,{k}^{2}+60\,{k}^{3}+36\,{k}^{4},  &\qquad |F_{12k+9}^{\si^l,\si^s}(F_4)|&=  4+25\,k+61\,{k}^{2}+72\,{k}^{3}+36\,{k}^{4}\\
|F_{12k+10}^{\si^l,\si^s}(F_4)|&= 8+42\,k+85\,{k}^{2}+84\,{k}^{3}+36\,{k}^{4},  &\qquad |F_{12k+11}^{\si^l,\si^s}(F_4)|&=10+52\,k+103\,{k}^{2}+96\,{k}^{3}+36\,{k}^{4}\\
|F_{12k}^{\si^l,\si^e}(F_4)|&=-2\,{k}^{2}+36\,{k}^{4}, &\qquad |F_{12k+1}^{\si^l,\si^e} (F_4)|&= -k-5\,{k}^{2}+12\,{k}^{3}+36\,{k}^{4}\\
|F_{12k+2}^{\si^l,\si^e}(F_4)|&= 4\,{k}^{2}+24\,{k}^{3}+36\,{k}^{4}, &\qquad |F_{12k+3}^{\si^l,\si^e}(F_4)|&= -k+7\,{k}^{2}+36\,{k}^{3}+36\,{k}^{4}\\
|F_{12k+4}^{\si^l,\si^e}(F_4)|&=4\,k+22\,{k}^{2}+48\,{k}^{3}+36\,{k}^{4},  &\qquad |F_{12k+5}^{\si^l,\si^e}(F_4)|&= 5\,k+31\,{k}^{2}+60\,{k}^{3}+36\,{k}^{4}\\
|F_{12k+6}^{\si^l,\si^e}(F_4)|&=2+16\,k+52\,{k}^{2}+72\,{k}^{3}+36\,{k}^{4}  &\qquad |F_{12k+7}^{\si^l,\si^e}(F_4)|&= 2+21\,k+67\,{k}^{2}+84\,{k}^{3}+36\,{k}^{4}\\
\end{alignedat}
\end{equation*}
\begin{equation*}
\begin{alignedat}{4}
|F_{12k+8}^{\si^l,\si^e}(F_4)|&= 6+40\,k+94\,{k}^{2}+96\,{k}^{3}+36\,{k}^{4},  &\qquad |F_{12k+9}^{\si^l,\si^e}(F_4)|&=  8+51\,k+115\,{k}^{2}+108\,{k}^{3}+36\,{k}^{4}\\
|F_{12k+10}^{\si^l,\si^e}(F_4)|&= 16+80\,k+148\,{k}^{2}+120\,{k}^{3}+36\,{k}^{4},  &\qquad |F_{12k+11}^{\si^l,\si^e}(F_4)|&=20+99\,k+175\,{k}^{2}+132\,{k}^{3}+36\,{k}^{4},\\
\end{alignedat}
\end{equation*}
\begin{equation*}
\begin{alignedat}{1}
|F_M^{\one,\si^l}(F_4)|&=|F_M^{\one,\si^s}(F_4)|,\\
|F_M^{\si^s,\si^l}(F_4)|&=|F_M^{\si^l,\si^s}(F_4)|.
\end{alignedat}
\end{equation*}}
\end{enumerate}

\end{document}